\documentclass[a4paper,10pt,oneside]{article}
\usepackage[a4paper, margin=1in]{geometry}
\usepackage{graphicx}
\usepackage{geometry}
\usepackage{natbib}
\usepackage{setspace}
\usepackage{hyperref}
\hypersetup{linkcolor=blue, urlcolor=blue}
\usepackage{amsfonts}
\usepackage{amssymb}
\usepackage{amsmath}
\usepackage{relsize}
\numberwithin{equation}{section}
\usepackage{amsthm}
\usepackage{soul}
\usepackage{mathtools}
\usepackage{multirow}
\usepackage{xcolor}
\usepackage{bm}
\usepackage{enumitem}
\newtheorem*{definition*}{Definition}

\newtheorem{lemma}{Lemma}
\newtheorem*{lemma*}{Lemma}

\newtheorem{example}{Example}

\def\defeq{:=}

\def\P{{\mathbb P}}     
\def\E{{\mathbb E}}
\def\T{{\mathbb T}}

\newcommand{\len}[2]{\ell(#1 \mid #2)}
\newcommand{\lensub}[2]{\ell^+(#1 \mid #2)}
\newcommand{\cost}[2]{c(#1 \mid #2)}
\newcommand{\gencost}[3]{#1(#2 \mid #3)}

\definecolor{gr}{rgb}{0.03, 0.7, 0.29}
\definecolor{Blue}{rgb}{0,0,1}

\begin{document}
	
	\title{Inconsistency of parsimony under the multispecies coalescent}
	\author{
		Daniel A. Rickert\textsuperscript{1,3,*}
		\and 
		Wai-Tong (Louis) Fan\textsuperscript{3,4,5}
		\and
		Matthew W. Hahn\textsuperscript{1,2}
	}
	
	\maketitle
	
	\vspace{-20pt}
	\begin{center}
		\textsuperscript{1}\textit{\small{Department of Biology, Indiana University, 1001 E 3rd St, 47405, Bloomington, IN, USA}}\\
		\textsuperscript{2}\textit{\small{Department of Computer Science, Indiana University, 700 N Woodlawn Ave, 47405, Bloomington, IN, USA}}
		\textsuperscript{3}\textit{\small{Department of Mathematics, Indiana University, 831 E 3rd St, 47405, Bloomington, IN, USA}}\\
		\textsuperscript{4}\textit{\small{School of Data Science and Society, University of North Carolina, 211 Manning Dr, 27514, Chapel Hill, NC, USA}}\\
		\textsuperscript{5}\textit{\small{Statistics and Operations Research, University of North Carolina, 318 E Cameron Ave, 27599, Chapel Hill, NC, USA}}\\
		\textsuperscript{*}\small{Corresponding author: daricker@iu.edu}
	\end{center}

	\begin{abstract}
		While it is known that parsimony can be statistically inconsistent under certain models of evolution due to high levels of homoplasy, the consistency of parsimony under the multispecies coalescent (MSC) is less well studied. Previous studies have shown the consistency of concatenated parsimony (parsimony applied to concatenated alignments) under the MSC for the rooted 4-taxa case under an infinite-sites model of mutation; on the other hand, other work has also established the inconsistency of concatenated parsimony for the unrooted 6-taxa case. These seemingly contradictory results suggest that concatenated parsimony may fail to be consistent for trees with more than 5 taxa, for all unrooted trees, or for some combination of the two. Here, we present a technique for computing the expected internal branch lengths of gene trees under the MSC. This technique allows us to determine the regions of the parameter space of the species tree under which concatenated parsimony fails for different numbers of taxa, for rooted or unrooted trees. We use our new approach to demonstrate that while parsimony succeeds in the unrooted 5-taxa case, there are regions of statistical inconsistency for concatenated parsimony for rooted 5+-taxa cases and unrooted 6+-taxa cases. Our results therefore suggest that parsimony is not generally dependable under the MSC.
	\end{abstract}
	
	\vspace{20pt}
	
	
	\maketitle
	
	\clearpage
	
	\section{Introduction}
	
	One of the major goals of phylogenetics is to describe the relationships among organisms. We suppose the evolutionary relationship among $n$ species or taxa can be described by a rooted, binary, and ultrametric species tree $\sigma = (\psi_*, \bm{x})$ with $n$ tips, where $\psi_*$ denotes the rooted binary topology of the species tree, and $\bm{x}$ gives the  branch lengths of $\sigma$. The goal is to be able to infer the species tree $\sigma$, or some component of it, such as the topology $\psi_*$, using data available from the tip species. 
	
	The most common data used to infer species trees come from DNA sequences. DNA sequences are available from every gene (or locus) in a genome. Coalescent-based models give a probability distribution on the gene tree $G$ that represents the evolutionary history of a given locus among sampled individuals \citep{kingman1982coalescent, hudson1990gene}. The gene tree topology, $G$, at a locus is conditionally random given the species tree, $\sigma$, when sampled individuals come from different species. The sequence data at this locus is then conditionally random given $G$, depending on any mutation events that have occurred on it. It is well-understood that the gene tree can be discordant (i.e. have internal branches that disagree) with the species tree for a number of biological reasons, such as introgression or horizontal gene transfer (see for instance \cite{maddison1997gene,edwards2009new}). However, arguably the most well-studied cause of gene tree discordance is incomplete lineage sorting (ILS), in which lineages in a population do not coalesce until entering a further ancestral population. In our analysis, we take ILS to be the sole cause of gene tree discordance, owing to the simplicity of mathematical models of ILS under the standard multispecies coalescent (MSC) model \citep{pamilo1988relationships, rannala2003bayes, rannala2020multispecies}. Recombination events along a chromosome allow neighboring loci to take on different gene tree topologies, all affected by the same biological processes. Accordingly, we assume that the gene tree $G$ at any given locus has a distribution given by the MSC for species tree $\sigma$. This distribution describes the probability distribution of the gene tree of a locus uniformly picked at random among a large number of loci.
	
	ILS is particularly common when the internal branch lengths of the species tree are short. In some regions of the parameter space of the species tree (called the anomaly zone, or AZ for short), a discordant rooted gene tree topology can be more likely to occur than one that matches the species tree topology \citep{Degnan2006}; a similar result also holds for unrooted gene tree topologies \cite{degnan2013anomalous}. Further work \citep{rosenberg2013discordance} has demonstrated that the AZ arises as the result of ILS on consecutive short branches of the species tree. Therefore, even in ideal world where we can infer gene tree topologies directly---essentially ignoring the randomness of sequences evolving on gene trees and the errors involved in inferring gene trees---the 'democratic vote' method that attempts to infer the species tree topology by simply returning the most common gene tree topology over many independent loci will be statistically inconsistent in some areas of parameter space \citep{degnan2009gene}, though more sophisticated methods using gene tree topology frequencies can give statistically consistent estimators of the species tree topology \citep{allman2011identifying, allman2018split}. The fact that the most probable gene tree topology under the MSC is not in general the species tree topology might appear to doom so-called concatenation methods, which combine data from multiple loci into a single alignment and essentially assume that all loci have evolved along the same gene tree. These concatenated alignments can be analyzed using a range of methods (e.g. parsimony, maximum likelihood, neighbor joining [NJ]) in order to return an estimated tree or topology. However, the regions of statistical consistency of such concatenation methods may differ from the AZ. For instance, simulations done in \cite{kubatko2007inconsistency} showed that, under the MSC, concatenated maximum-likelihood (ML) for 4 taxa could be consistent inside the AZ and inconsistent outside of it. This was shown more exhaustively in \cite{MendesHahn2018}, who sampled a far greater number of points in parameter space. 
	
	Perhaps surprisingly, \cite{liu2009phylogenetic} and \cite{MendesHahn2018} found that concatenated parsimony was statistically consistent for the rooted 4-taxa case across parameter space, assuming an infinite-sites mutation model and the MSC model. These findings contrast with the well-known results described in \cite{felsenstein1978cases}, which found an area of parameter space of statistical inconsistency of parsimony (sometimes called the Felsenstein zone). These two sets of results do not conflict, as inconsistency in the Felsenstein zone is caused by similarity due to homoplasy (multiple substitutions at a site), a phenomenon that does not occur in the infinite-sites model. It should also be noted that the analysis in \cite{felsenstein1978cases} does not incorporate gene tree discordance, while the results of \cite{liu2009phylogenetic} and \cite{MendesHahn2018} do. 
	
	In this work, we focus on concatenated parsimony and similar concatenation methods ('concatenated counting methods'), all of which take a concatenated alignment, $\mathcal{A}$, as input, and associate a 'cost' $c(\psi \mid \chi)$ to each candidate topology, $\psi$, and site pattern, $\chi \in \mathcal{A}$. These methods then attempt to infer the species tree topology by returning the candidate topology $\psi$ that minimizes the total cost across the entire concatenated alignment. Such methods are similar to the idea behind concatenated ML, with the difference that concatenated ML attempts to minimize the total negative log-likelihood of a candidate tree (with branch lengths included) rather than just the total cost of a candidate topology (which encodes no information about branch lengths). The idea of cost minimization is also common in various gene tree methods using a collection, $\mathcal{G}$, of (estimated) gene trees and a cost function, $c(\psi \mid G)$, for each $G \in \mathcal{G}$. For instance, taking a cost function representing whether the rooted topology of $G$ matches with $\psi$ gives the 'democratic vote' method, while taking a cost function that returns the number of shared quartets between $\psi$ and the unrooted topology of $G$ motivates ASTRAL \citep{mirarab2014astral}. Other choices of cost function have also been examined, for instance the minimize-deep-coalescence (MDC) criterion \citep{maddison1997gene, maddison2006inferring, than2011consistency}. 
	
	In the same manner that examining the statistical consistency of gene tree methods as more loci are sampled usually requires calculating the frequencies of gene tree topologies under the MSC, examining the statistical consistency of concatenated counting methods involves calculating the expected lengths of branches of gene trees under the MSC. We begin by presenting a novel combinatoric technique to calculate these expected lengths, and demonstrate that the technique correctly recovers known results in the 4-taxa case. We then apply this technique to further understand the success and failure of concatenated parsimony methods for cases with 5 or more taxa (``5+ taxa''), both in the rooted and unrooted cases. While \cite{roch2015likelihood} have demonstrated inconsistency of concatenated parsimony for the unrooted 6-taxa case under a general $r$-state mutation model (even when homoplasy is negligible), their results do not characterize the precise regions of parameter space where parsimony fails; instead, their model assumed the probability of coalescence in internal branches is sufficiently small, justifying the use of Ewens' sampling formula \citep{ewens1972sampling} in computing the probabilities of site patterns. Moreover, \cite{bryanthahn2020} have argued that the results of \cite{roch2015likelihood} only directly demonstrate inconsistency for biologically unrealistic species tree branch lengths. With our method of computing expected branch lengths, we demonstrate that for the previously unexplored unrooted 5-taxa case, concatenated parsimony is consistent under a MSC + infinite-sites model of evolution. We also show that under the same modeling assumptions, concatenated parsimony always has a region of inconsistency for the rooted 5+-taxa case and the unrooted 6+-taxa case, and find that this anomalous region is non-trivial, including many biologically realistic species trees.  We conclude by discussing the implications of our results for the accurate inference of species trees. 
	
	\section{Definitions} \label{S:def_result}
	
	Let $\T_n$ denote the set of all rooted, binary, and labeled tree topologies on $n$ taxa, with tips labeled by the label set $[n] =\{1,2, \ldots, n \}$. For clarity, we will often use uppercase letters $\{A,B,C,\ldots,\}$ as the label set in place of $[n]$ when discussing specific examples. For each rooted topology $\psi \in \T_n$, we let $\overline{\psi}$ be its unrooted analogue, and define $\overline{\T}_n$ to be the collection of all such unrooted $n$-taxa topologies. 
	
	We think of a rooted, binary, and ultrametric species tree with $n$ taxa as a pair $\sigma = (\psi_*, \bm{x})$, where $\psi_* \in \T_n$ denotes the rooted topology of $\sigma$ and $\bm{x}$ is a vector of non-negative branch lengths of the species tree. For convenience, we will assume each branch of the species tree has a constant, large effective population size of $N_e$ diploid individuals, and that all time and lengths are measured in coalescent units of $2N_e$ generations. For our data, we will assume an MSC + infinite-sites model of evolution, as follows:
	\begin{itemize}
		\item Loci are labeled by $i \in \{1,2, \ldots \} $, with the gene tree $G_i$ at locus $i$ being independently and identically distributed according to the multispecies coalescent (MSC) on the species tree $\sigma$, with one sequence sampled per taxa. In particular, any pair of lineages of the gene tree that exist in the same ancestral population coalesce independently at rate $1$. We think of each gene (or lineage) as being labeled by the label of the taxa it is sampled from. For example, the collection of lineages in the set $\{A,B,C\}$ refers to the lineages of the gene tree that are sampled from taxa $A,B,C$.   
		\item Each locus consists of infinitely many sites, and mutations fall on the gene tree $G_i$ according to an infinite sites model with rate $\theta/2$, where $\theta=4N_e\mu$ is the scaled mutation rate, and $\mu$ is the per-generation mutation rate per locus, assumed to be constant throughout time and across loci. 
		\item The alignment at locus $i$ is denoted $A_i$. The $j^{\rm th}$ row of $A_i$ corresponds to the sequence of the sampled gene from taxa $j$, and each column corresponds to the nucleotide data for each of the $n$ samples at a particular site.
		See Figure \ref{fig:concat}a for an illustration.
		\item The ancestral allelic state (i.e. the allelic state immediately prior to a mutation event) at a site is denoted by $0$, and the derived allelic state at a site is denoted by $1$. 
		\item Each site in the alignment is summarized by a site pattern $\chi \subseteq [n]$ denoting the sequences at that site that have the derived allelic type $1$. We will often write a site pattern by a concatenated string of its elements, for example $\chi = AB$ in place of $\chi = \{A, B\}$. 
	\end{itemize}
	Assuming that we can distinguish between the derived and ancestral allelic states at each site, we may attempt inference of the rooted species tree topology $\psi_*$. A \textit{rooted concatenated counting method} assigns a cost $c(\psi \mid \chi)$ to each pair of rooted candidate topology $\psi \in \T_n$ and site pattern $\chi$. The total cost of the candidate topology $\psi$ on an concatenated alignment $\mathcal{A}^{(k)}$, obtained by combining alignments $A_1, \ldots A_k$ (Figure \ref{fig:concat}b), is given by
	\begin{equation} \label{eq:tot_cost}
		\mathrm{Cost}(\psi \mid \mathcal{A}^{(k)}) = \sum_{\chi} c(\psi \mid \chi) \#(\chi \mid \mathcal{A}^{(k})
	\end{equation}
	where the sum is over all $\chi \subseteq [n]$, and $\#(\chi \mid \mathcal{A}^{(k)})$ is the number of occurrences of the site pattern $\chi$ in $\mathcal{A}^{(k)}$. The resulting estimator is given by the topology $\psi$ that minimizes the total cost; if there is a tie between multiple elements of $\T_n$, then we pick one uniformly at random. 
	\begin{figure}[h]
		\centering
		\sbox0{\includegraphics{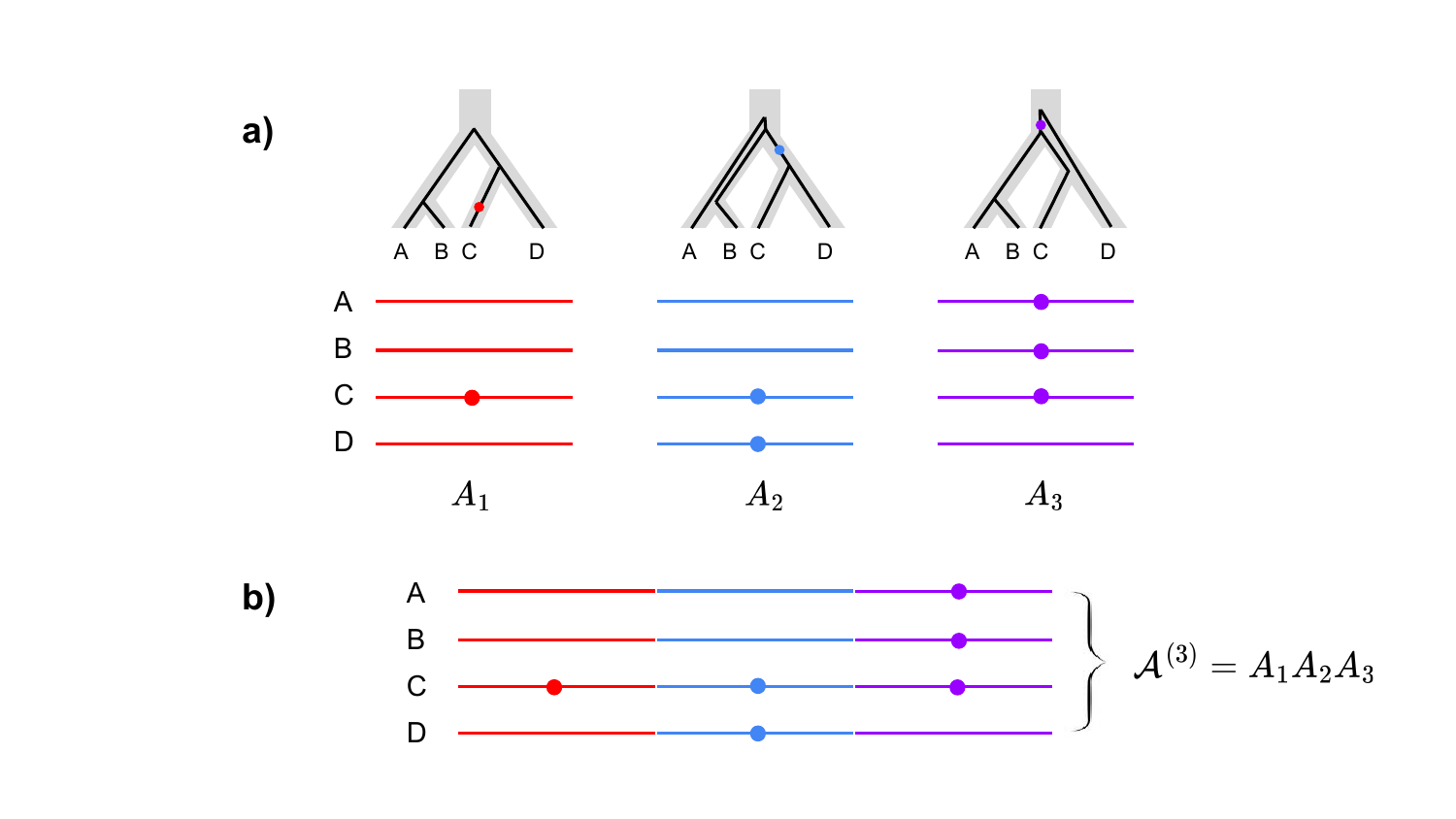}}%
		\includegraphics[trim={.1\wd0} 0 0 {.02\wd0}, clip, width=18cm]{Figure1.pdf}
		\caption{Illustration of the concatenation procedure. a) Three gene trees $G_1, G_2, G_3$ (black lines) and alignments $A_1, A_2, A_3$ for $k=3$ independently evolving loci on the same species tree (shaded gray). Each mutation (circle) represents a $0 \to 1$ (ancestral to derived) transition at a new site in the locus. Only one segregating site per locus is shown for simplicity. b) The resulting concatenated alignment $\mathcal{A}^{(3)}=A_1 A_2 A_3$, with three segregating sites.}\label{fig:concat}
	\end{figure}
	
	On the other hand, when it is not possible to distinguish between the derived and ancestral allelic states, we will need to arbitrarily assign the derived and allelic states at each site. Hence, the site pattern $\chi$ and its complement $[n] \setminus \chi$, resulting from the two possible choices of assignment, must be treated equally in any method of inference. Usually, this limits us to attempting inference of the unrooted species tree topology only. An \textit{unrooted concatenated counting method} assigns a cost $c(\overline{\psi} \mid \chi)$ to each pair of unrooted candidate topology $\overline{\psi} \in \T_n$ and site pattern $\chi$, with the symmetry requirement $c(\overline{\psi} \mid \chi) = c(\overline{\psi} \mid [n] \setminus \chi)$, and attempts to minimize 
	\begin{equation} \label{eq:tot_cost_unrooted}
		\mathrm{Cost}(\overline{\psi} \mid \mathcal{A}^{(k)}) = \sum_{\chi} c(\overline{\psi} \mid \chi) \#(\chi \mid \mathcal{A}^{(k})
	\end{equation}
	For both rooted and unrooted concatenated counting methods, we will also make the following regularity assumptions on the choice of costs, which are met by the all the concatenated concatenated methods we discuss. 
	\begin{description}[style=unboxed,leftmargin=0cm]
		\item[Parsimony uninformative sites are ignored] A site with site pattern $\chi$ is said to be parsimony informative (or simply informative) if $2 \leq |\chi| \leq n-1$ in the rooted case, and if $2 \leq |\chi| \leq n-2$ in the unrooted case. All other sites/site patterns are considered uninformative and are excluded when calculating the total cost as in \eqref{eq:tot_cost}. See Figure \ref{fig:isp} for an illustration of the distinction between an informative and uninformative site pattern.
		
		\textit{Implication}: Informative sites arise as the result of mutations occurring on internal branches of the rooted/unrooted gene tree. Since no coalescent events on the gene trees $(G_i)_i$ occur on external branches of the species tree, the lengths of these external branches are irrelevant to the behavior of the estimator. Therefore, we will write the species tree branch lengths $\bm{x}=(x_1, \ldots, x_{n-2})$ to be a vector of $n-2$ internal branch lengths throughout the paper.
		
		\item[Exchangeability of labels of taxa] The cost function should not favor any particular taxa or group of taxa over any other based solely on their (arbitrary) labeling. In particular, if $\pi$ is any permutation of $[n]$, we require that $\cost{\pi(\psi)}{\chi(\pi)} = \cost{\psi}{\chi}$, where $\pi(\psi) \in \T_n$ is the topology obtained by permuting the labels of the tips of $\psi$ by $\pi$, and $\pi(\chi) \subseteq [n]$ is the site pattern obtained by applying $\pi$ to each of the elements of $\chi$.  
		
		\textit{Implication}: The behavior of a concatenated counting method does not depend on the choice of labeling. Therefore, in any exploration of species tree space, it suffices to examine one representative of each possible unlabeled $n$-taxa topology, instead of examining all possible labeled $n$-taxa topologies.
		
	\end{description}

	\begin{figure}[h]
		\centering
		\includegraphics[width=12cm]{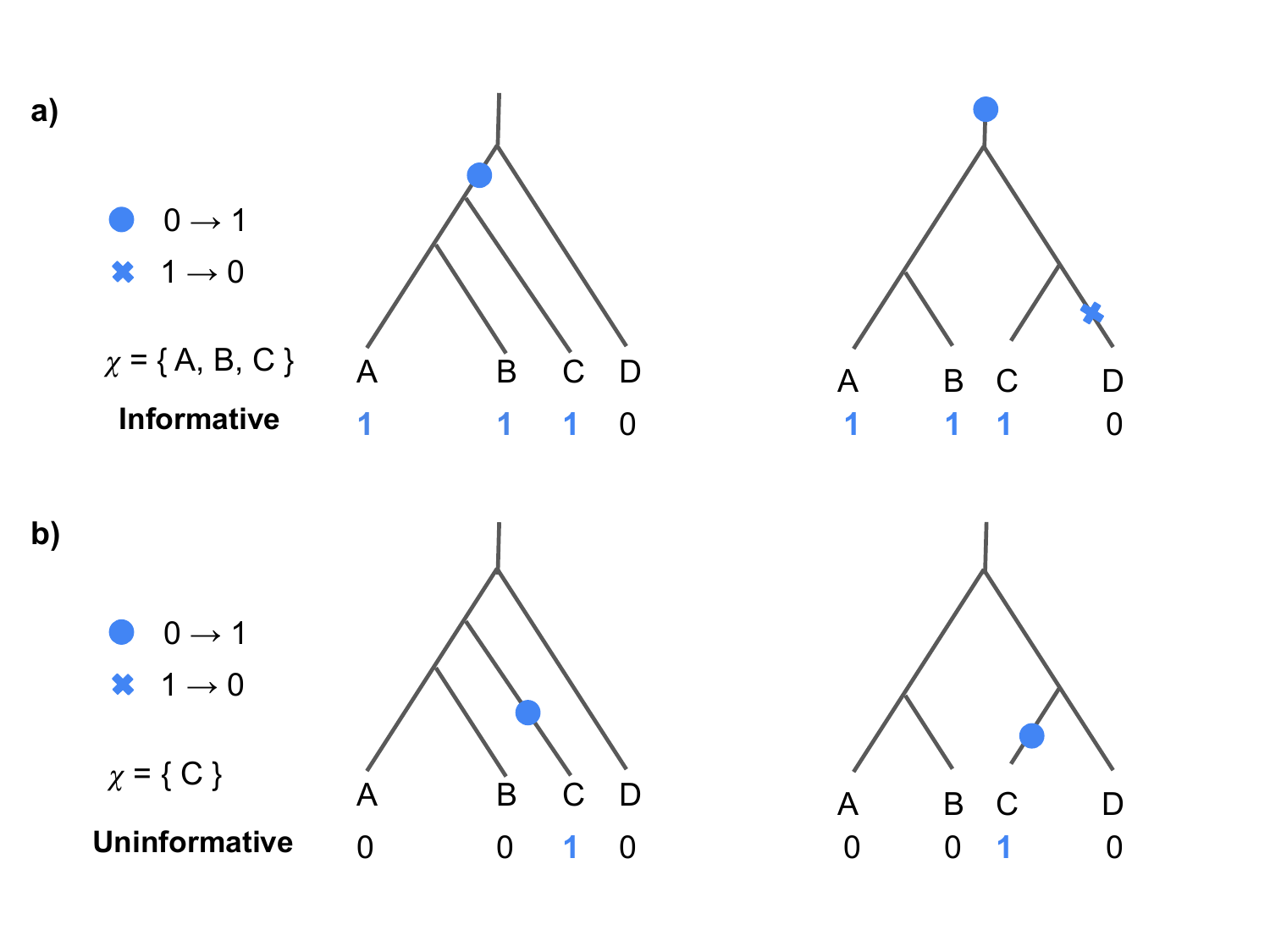}
		\caption{a) An informative site pattern (ISP), $\chi=\{A,B,C\}$ resolved on two different candidate tree topologies requiring different numbers of mutations; b) an uninformative site pattern, $\chi=\{C\}$ resolved on two different topologies requiring only one mutation in both cases.}\label{fig:isp}
	\end{figure}
	
	The most well-known group of concatenated counting methods are parsimony methods, such as Camin-Sokal parsimony \citep{camin1965method}, Wagner (unordered) parsimony \citep{kluge1969quantitative, farris1970methods}, and Dollo parsimony \citep{le1974uniquely, le1977uniquely, farris1977phylogenetic}. See \cite{Felsenstein1983} for an overview of these methods. Each of these methods assigns costs $\cost{\psi}{\chi}$ by counting the minimal number of mutations needed to explain a site pattern $\chi$ on the topology $\psi$, with variation among the methods due to different assumptions on the allowed transitions from the ancestral allelic state $(0)$ to the derived allelic state $(1)$. Among these parsimony methods, we will focus our attention on Wagner parsimony (henceforth just referred to as parsimony), so we give a quick definition here:
	
	\begin{description}[style=unboxed,leftmargin=0cm]
		\item[Rooted parsimony] The cost $\cost{\psi}{\chi}$ is the minimum total number of $0 \to 1$ and $1 \to 0$ mutations needed to resolve $\chi$ on the topology $\psi$, assuming an ancestral state of $0$. By convention, a $0 \to 1$ mutation may be placed above the root of the topology. See Figure \ref{fig:isp}a for an example.  
		
		\item[Unrooted parsimony] The cost $\cost{\overline{\psi}}{\chi}$ is the minimum total number of $0 \to 1$ and $1 \to 0$ mutations needed to resolve $\chi$ on a rooted representative $\psi$ of $\overline{\psi}$, allowing the ancestral state to be $0$ or $1$.  
	\end{description}
	Parsimony methods are not the only concatenated counting methods, with recent work highlighting quartet-based approaches, both in the case of an infinite-sites model of mutation (e.g. SDPQuartets / ASTRAL-BP \citep{springer2020ils, molloy2022theoretical}) and more general mutation models (e.g. CASTER \citep{zhang2025caster}). Both SDPQuartets and ASTRAL-BP are closely related to parsimony -- using the fact that parsimony is statistically consistent for the unrooted 4-taxa case under a MSC + infinite-sites model of mutation \citep{MendesHahn2018, molloy2022theoretical} -- though the exact implementation used to search for the optimal candidate tree topology varies between these methods. CASTER is also closely related to parsimony, but involves a negative weighting of some parsimony-uninformative sites to compensate for the possibility of multiple mutations at a single site. Since our focus is on an infinite-sites model of mutation, we give a brief framing of the methods SDPQuartets / ASTRAL-BP as concatenated counting methods. 
	
	For any four distinct elements $a,b,c,d$ of the label set $[n]$, we say that unrooted topology $\overline{\psi} \in \overline{\T}_n$ displays the quartet $ab|cd$ if the restriction of $\psi$ to taxa $a,b,c,d$ has unrooted topology $((ab)(cd))$, i.e. taxa $a, b$ are sisters in $\overline{\psi}$ as are taxa $c,d$. Similarly, we say that a site pattern $\chi$ supports the quartet $ab|cd$ if $a,b \in \chi$ but $c,d \notin \chi$ or vice versa. 
	\begin{description}[style=unboxed,leftmargin=0cm]
		\item[SDPQuartets/ASTRAL-BP] The cost $\cost{\overline{\psi}}{\chi}$ is taken to be $-\gencost{q}{\overline{\psi}}{\chi}$, where $\gencost{q}{\overline{\psi}}{\chi}$ is the number of the $\binom{n}{4}$ unrooted quartets $ab|cd$ implied by $\psi$ that the site pattern $\chi$ supports. We call $\gencost{q}{\psi}{\chi}$ the quartet score. 
	\end{description}
	
	To examine the statistical consistency of concatenated counting methods as the number of sampled loci included in the alignment grows large, we note that the expected number of sites with site pattern $\chi$ in the alignment $A_i$ is proportional to the expected length of the  branch in $G_i$ which subtends exactly the lineages in $\chi$. Hence, the expected contribution of the alignment $A_i$ to the total cost in \eqref{eq:tot_cost} is proportional to 
	\begin{equation}
		C(\psi \mid \sigma) := \sum_{\chi} \cost{\psi}{\chi} \len{\chi}{\sigma}
	\end{equation}
	where $\len{\chi}{\sigma}$ is the expected branch length of a branch subtending exactly the genes in $\chi$ for a gene tree generated under the MSC on $\sigma$. One may analogously define the expected cost per locus $C(\overline{\psi} \mid \sigma)$ in the unrooted case. Therefore, the strong law of large numbers gives the following criteria for consistency/inconsistency of the given concatenated counting estimator under the species tree $\sigma$ (with an analogous criterion holding in the unrooted case): 
	\begin{description}
		\centering
		\item[(Consistency)] $C(\psi_* \mid \sigma) < C(\psi \mid \sigma)$ for all $\psi \neq \psi_* \in \T_n$.
		\item[(Inconsistency)] $C(\psi_* \mid \sigma) \geq C(\psi \mid \sigma)$ for some $\psi \neq \psi_* \in \T_n$.
	\end{description}
	The main objectives of this paper are to present a novel method of computing the quantity $\len{\chi}{\sigma}$, and analyze the statistical consistency of concatenated counting methods (in particular parsimony) as follows:
	\begin{enumerate}
		\item In Section 3.1, we demonstrate a new, simple manner of computing expected branch lengths $\len{\chi}{\sigma}$ in the gene tree under the MSC.
		\item In Section 3.2, we discuss a decomposition of expected branch lengths  $\len{\chi}{\sigma}$ and related quantities, which is useful both conceptually and computationally.    
		\item In Section 3.3, we provide sample calculations of finding $\len{\chi}{\sigma}$ in the 4-taxa case, and show that these results agree with existing work;  
		\item In Section 4.1, we prove the statistical consistency of parsimony under a MSC + infinite-sites model of evolution in the previously analyzed unrooted 5-taxa case;
		\item In Sections 4.2 and 4.3, we show that parsimony is statistically inconsistent under a MSC + infinite-sites model of evolution in the rooted 5-taxa and unrooted 6-taxa cases, and characterize the exact regions of statistical consistency. 
	\end{enumerate}
	
	\section{Expected branch lengths under the MSC}
	In Section \ref{S:def_result}, we saw that establishing the statistical (in)consistency of concatenated counting methods requires calculating the expected branch lengths of gene trees. In this section, we demonstrate how existing work on the expected height of gene trees under the MSC, in particular that of \cite{efromovich2008coalescent}, may be used to compute the expected branch lengths in gene trees under the MSC. The idea of using the the expected height of the MRCA of sampled tips (usually, by sampling multiple tips per taxa) has been used more directly in species tree inference by methods such as GLASS \citep{mossel2008incomplete}, iGLASS \citep{jewett2012iglass}, STEAC \citep{liu2009estimating}, and MAC \citep{helmkamp2012improvements}; however, here we will use the expected heights as a means to an end. We note that some calculations regarding expected branch lengths have already been done in the 3- and 4-taxa cases, though the approach taken by existing methods makes extensions to $5+$-taxa difficult. For instance, \cite{MendesHahn2018} did so by computing the expected branch length conditional on each possible gene tree history, a method that quickly becomes infeasible for large $n$ since there are $(2n-3)!! = (2n-3) \times (2n-1) \times \cdots \times 1$ possible gene tree topologies and even more possible gene tree histories. Alternatively, a diffusion approximation has been used \citep{doronina2017speciation}, but this approach requires similar amounts of work.   Our new approach involves a relatively simple combinatoric trick that minimizes individual calculations and enables easier analysis.

	\subsection{Gene tree lengths, subtending lengths, and heights}
	\noindent We first introduce notation for the length of a branch and the extended length of a branch in a (random) gene tree $G$. We will not continue the gene tree above its root (i.e. the MRCA of all sampled lineages), since mutations that occur above the root do not give rise to parsimony informative sites.  
	
	\begin{definition*}[Length of a branch] 
		For $\chi \subseteq [n]$, the random variable $L(\chi)$ is the length of the branch subtending exactly the lineages in $\chi$ in the random gene tree $G$ if such a branch exists; otherwise it is defined to be $0$. 
	\end{definition*}
	
	\begin{definition*}[Extended length of a branch]  
		For $\chi \subseteq [n]$, the random variable $L^+(\chi)$ is the total length of branches that subtend at least the lineages in $\chi$ in the random gene tree $G$:
		\begin{equation} \label{eq:extended_length}
			L^+(\chi) = \sum_{\eta \supseteq \chi} L(\eta).
		\end{equation}
	\end{definition*}
	\noindent Ignoring branch lengths that are zero in the sum $\sum_{\eta \supseteq \chi} L(\eta)$, we see that the extended length $L^+(\chi)$ amounts to the length of the path from the MRCA of the lineages in $\chi$ to the root of the gene tree. In particular,
	\begin{equation*}
		L^+(\chi) = H([n]) - H(\chi)
	\end{equation*}
	where $H([n])$ is the height of the MRCA of all $n$ lineages, and $H(\chi)$ is the height of the MRCA of lineages in $\chi$. An example of the relationships connecting $L(\chi), L^+(\chi)$ and $H(\chi)$ is given in Figure \ref{fig:subtending_length_examples}. By taking expected values,
	\begin{equation} \label{eq:exp_subtending}
		\ell^+(\chi \mid \sigma) \defeq  \E_\sigma[L^+(\chi)] = h([n] \mid \sigma) - h(\chi \mid \sigma),
	\end{equation}
	where $h([n] \mid \sigma)$ is defined to be the expected height of a gene tree with one sample from each taxon, and $h(\chi \mid \sigma)$ is defined to be the expected height of a gene tree with one sample from each of the taxa in $\chi$ only. Both expectations are under the MSC on $\sigma$. In computing $h(\chi \mid \sigma)$, we may restrict $\sigma$ to just the taxa in $\chi$, since we only sample lineages from these taxa. The computation of these heights may be done by a standard dynamic programming method; see for instance \citep[Eq. (6)]{efromovich2008coalescent}.
	\begin{figure}[ht]
		\centering
		\sbox0{\includegraphics{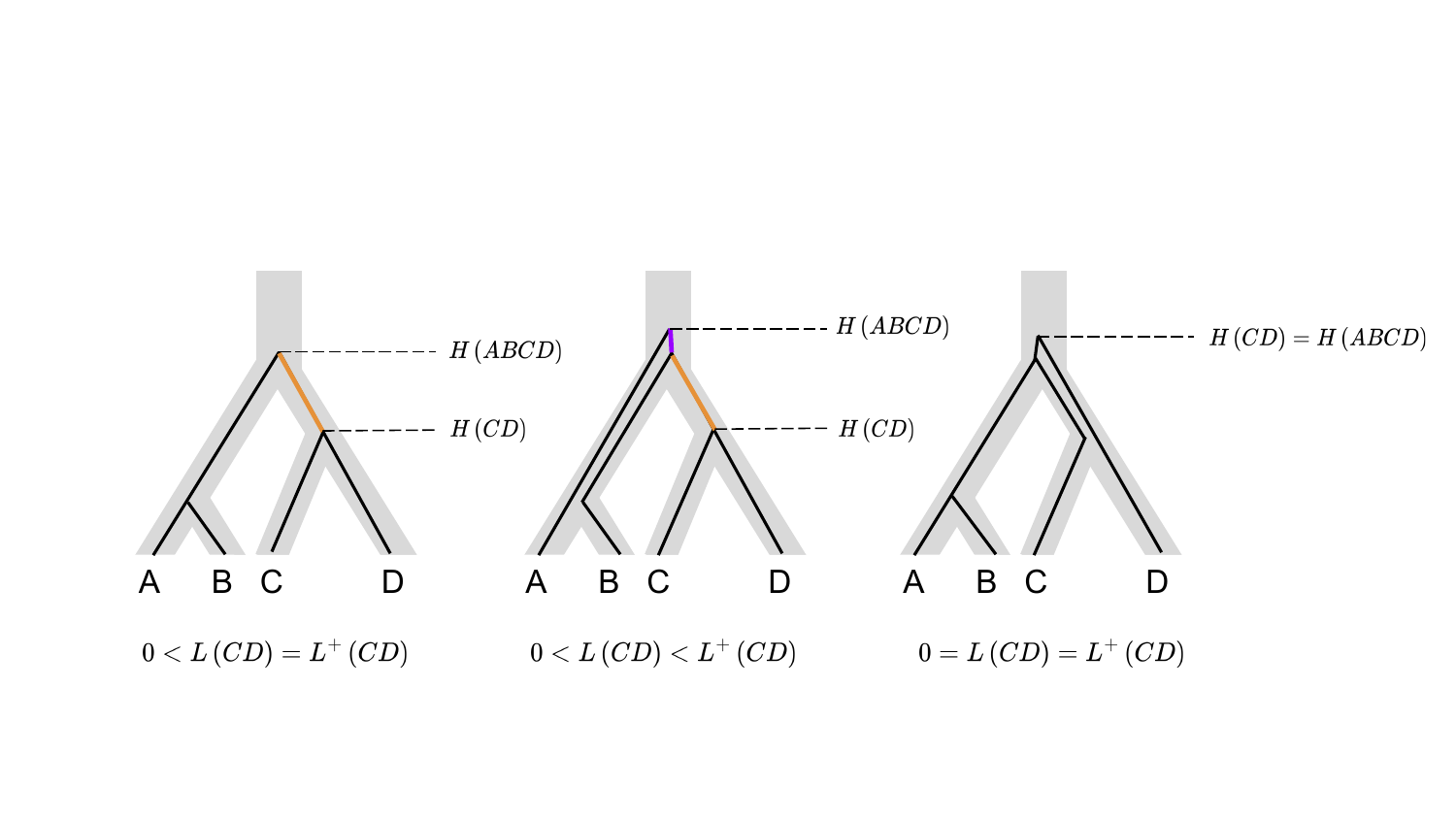}}
		\includegraphics[trim={.03\wd0} {.1\wd0} 0 {.2\wd0}, clip, width=15cm]{Figure3.pdf}
		\caption{Three different realizations of a gene tree $G$ within a species tree with topology $\psi_* = ((AB)(CD))$. Branches of the gene tree that contribute to the length $L(CD)$ and the subtending length $L^+(CD)$ are highlighted in orange, whereas branches of the gene tree that contribute to $L^+(CD)$ but not $L(CD)$ are highlighted in purple. The heights of the overall gene tree $H(ABCD)$ and the height $H(CD)$ of the gene tree restricted to tips $C,D$ are also shown.}\label{fig:subtending_length_examples}
	\end{figure}
	
	To get a feel for how the lengths of branches can be related back to the extended lengths of branches, it can be helpful to make a connection with indicator functions and the inclusion-exclusion principle. Consider the four-taxa case with a  site pattern $\chi = \{a, b\}$, letting $c, d$ denote the other two taxa of the species tree not in $\chi$. Then we can write
	\begin{align*}
		L(ab) &= \sum_{\chi \subseteq \{a,b,c,d\}} 1\{a, b \in \chi, c,d \notin \chi\} L(\chi) \\
		L^+(ab)  &=  \sum_{\chi \subseteq \{a,b,c,d\}} 1\{a, b \in \chi \} L(\chi)
	\end{align*}
	and similarly for other site patterns. We then expand the indicator function $1\{a, b \in \chi, c,d \notin \chi\}$:
	\begin{align*}
		& \, 1\{a, b \in \chi, c, d \notin \chi\} \\
		=& \, 1\{a, b \in \chi\} \cdot 1\{c \notin \chi\} \cdot 1\{d \notin \chi\} \\ 
		=& \, 1\{a, b \in \chi\}\cdot \big(1-1\{c \in \chi\}\big) \cdot \big(1-1\{d \in \chi\}\big) \\ 
		=& \, 1\{a, b \in \chi\} - 1\{a, b, c \in \chi\} -  1\{a, b, d \in \chi\} +  1\{a, b, c, d \in \chi\}
	\end{align*}
	The terms on the last line, after multiplying by $L(\chi)$ and summing over all $\chi \subseteq \{a,b,c,d\}$, correspond exactly to $+L^+(ab)$, $-L^+(abc)$, $-L^+(abc)$, and $+L^+(abcd)$ respectively. (Note that $L^+(abcd)$ is $0$ in this example, since no internal branch of $G$ subtends all four sampled lineages). This gives an idea for how we may 'invert' the relationship given in \eqref{eq:extended_length}; also see Figure \ref{fig:ie_trees} for a further visual example of this inversion in the $5$-taxa case. 
	
	\begin{figure}[h]
		\centering
		\includegraphics[width=16cm]{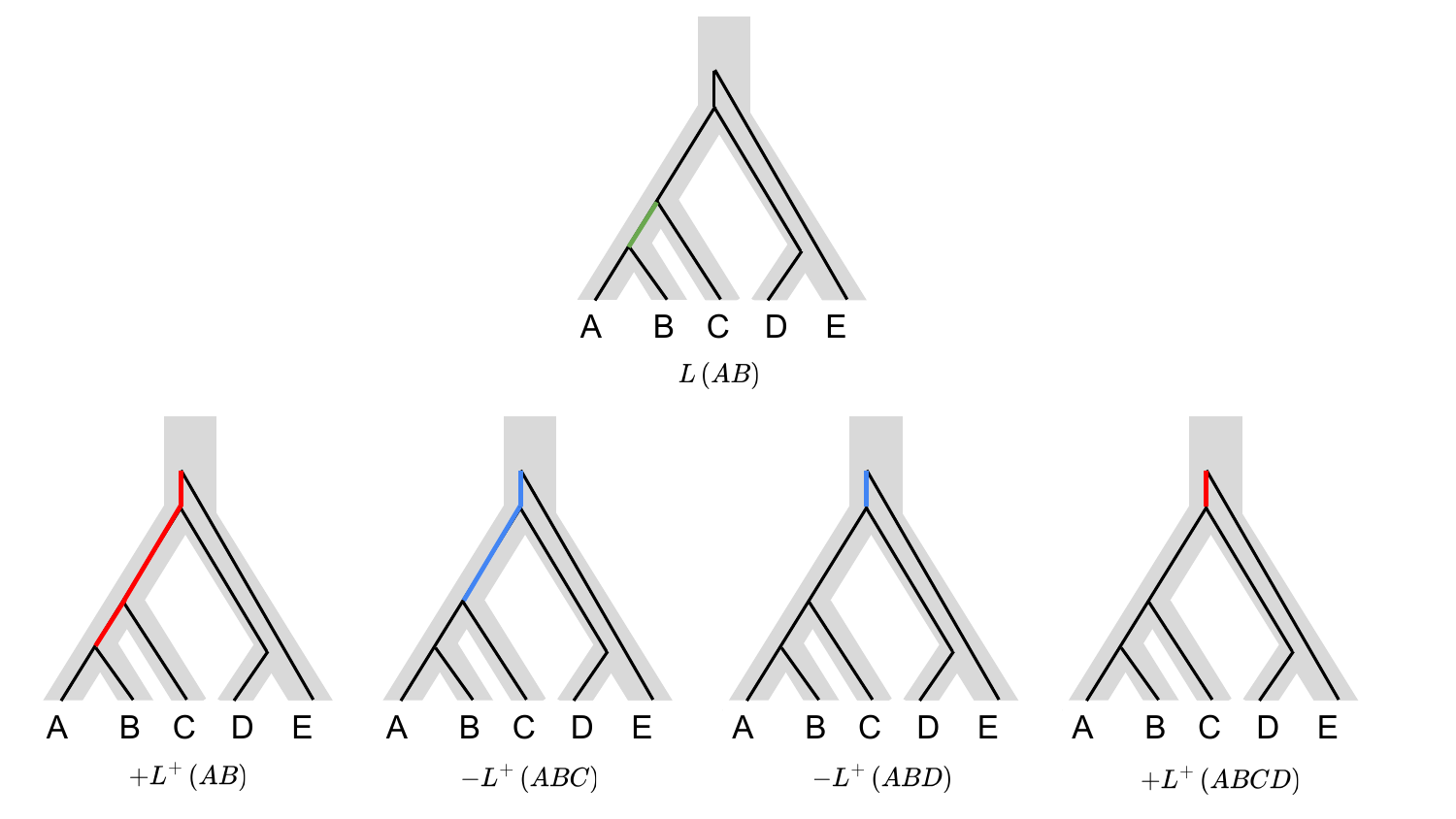}
		\caption{Illustration of recovering the length $L(\chi)$ for $\chi=AB$ from the collection of extended lengths $(L^+(\eta))_{\eta \supset \chi}$ on a gene tree. Terms with an extended length $L^+(\eta) = 0$ for this realization of the gene tree, namely those with $E \in \eta$, are omitted.}\label{fig:ie_trees}
	\end{figure}
	
	\begin{lemma} \label{lemma:inversion} 
		For any fixed realization of the gene tree, 
		\begin{equation*}
			L(\chi) = \sum_{\eta \supseteq \chi} (-1)^{|\eta|-|\chi|} L^+(\eta). 
		\end{equation*}
	\end{lemma}
	
	\begin{proof}
		The claim follows from a reversal of the order of summation:
		\begin{align*}
			\sum_{\eta \supseteq \chi} (-1)^{|\eta|-|\chi|} L^+(\eta) &=  \sum_{\eta \supseteq \chi} (-1)^{|\eta|-|\chi|} \sum_{\eta' \supseteq \eta} L(\eta') \\
			&= \sum_{\eta' \supseteq \chi} L(\eta') \sum_{\chi \subseteq \eta \subseteq \eta'}  (-1)^{|\eta|-|\chi|}  \\ 
			&=  \sum_{\eta' \supseteq \chi} L(\eta') 1_{\{\eta' = \chi\}} \\
			&= L(\chi)
		\end{align*}
		Note the third line follows since whenever $\eta'$ strictly contains $\chi$, there are an equal number of sets $\eta$ with even parity and odd parity lying between $\chi$ and $\eta'$ (which can be seen by applying the binomial theorem to $(1 + (-1))^{|\eta'|-|\chi|}$), such that the sum $\sum_{\chi \subseteq \eta \subseteq \eta'}  (-1)^{|\eta|-|\chi|}$ vanishes.
	\end{proof}
	\noindent By taking expected values in Lemma \ref{lemma:inversion}, we obtain a practical method of computing $\len{\chi}{\sigma}$ in general:
	\begin{equation} \label{eq:inversion}
		\len{\chi}{\sigma} = \sum_{\eta \supseteq \chi} (-1)^{|\eta|-|\chi|}\lensub{\chi}{\sigma}
	\end{equation}
	or, alternatively, after applying \eqref{eq:exp_subtending} and canceling terms $h([n] \mid \sigma)$, 
	\begin{equation} \label{eq:inversion2}
		\len{\chi}{\sigma} = -\sum_{\eta \supseteq \chi} (-1)^{|\eta|-|\chi|} h(\chi \mid \sigma)
	\end{equation}
	
	The generality of the above ideas connecting the expected length $\len{\chi}{\sigma}$ to the expected heights of gene trees should be noted. For instance, if we extend the algorithm of \cite{efromovich2008coalescent} to calculate expected gene tree heights for a given phylogenetic network (i.e. in models involving introgression), we can still use \eqref{eq:inversion} to analyze the regions of inconsistency for a given concatenated counting method (cf. Hibbins and Hahn 2024). It could also be applied more broadly to find expected branch lengths in random coalescent trees within a single population, even when the coalescent tree admits multiple mergers, such as in the $\Lambda$-coalescent \citep{pitman1999coalescents} or the $\Xi$-coalescent \citep{schweinsberg2001coalescents}, and this idea has been explored in \cite{spence2016site}, though in an indirect manner that did not fully explore the underlying connection to the inclusion-exclusion principle. 
	
	It is also possible that a similar combinatoric approach could be applied to find the probability that a randomly sampled gene tree has a particular branch (or collection of branches), rather than to find the expected lengths of branches. This may allow alternative derivations of results on monophyly and paraphyly as examined in several previous works \citep{rosenberg2003shapes, mehta2016probability, mehta2019probability, mehta2022probability}.

	\subsection{Decomposition into species and coalescent terms} \label{SS:st}
	\noindent In performing analysis with the expected branch lengths, it can be helpful to first decompose the expected heights $h(\chi \mid \sigma)$ as the height $h_{\rm sp}(\chi \mid \sigma)$ of the species MRCA of the taxa in $\chi$ in the species tree $\sigma$, plus the expected height $h_{\rm coa}(\chi \mid \sigma)$ of the MRCA of the lineages in $\chi$ above the species MRCA of the taxa in $\chi$. We will call these terms the 'species' and 'coalescence' terms respectively; \cite{efromovich2008coalescent} refers to the coalescence term as the 'species-gene coalescent time'.  
	\begin{equation*}
		h(\chi \mid \sigma) = h_{\rm sp}(\chi \mid \sigma) + h_{\rm coa}(\chi \mid \sigma)
	\end{equation*}
	The species term $ h_{\rm sp}(\chi \mid \sigma)$ is the sum of a collection of branch lengths of $\sigma$. Meanwhile the coalescent term $h_{\rm coa}(\chi \mid \sigma)$ is a polynomial in the transformed branch lengths $X_i:=\exp(-x_i)$, since the coalescent transition probabilities are a polynomial function of $\exp(-t)$ \citep{Tavare1984}.
	
	By grouping any species and coalescence terms that appear together in an expression, we can similarly define a decomposition for terms such as the expected subtending length $\ell^+(\chi \mid \sigma)$, the expected length $\ell(\chi \mid \sigma)$, and the expected cost per locus $C(\psi \mid \sigma)$. For instance, in the decomposition 
	\begin{equation} \label{eq:len_decomp}
		\len{\chi}{\sigma} = \ell_{\rm sp}(\chi \mid \sigma) + \ell_{\rm coa}(\chi \mid \sigma) 
	\end{equation}
	the species term  $\ell_{\rm sp}(\chi \mid \sigma)$ is nothing but the length of the branch in $\sigma$ that subtends exactly the tips in $\chi$, if such a branch in $\sigma$ exists; otherwise it is $0$. It is also useful to write the decomposition for the expected cost per locus:
	\begin{align} \label{eq:cost_decomp}
		C(\psi \mid \sigma) = \underset{C_{\rm sp} (\psi \mid \sigma)}{\underbrace{\sum_{\chi} \cost{\psi}{\chi} \ell_{\rm sp}(\chi \mid \sigma)}} + \underset{ C_{\rm coa} (\psi \mid \sigma) }{\underbrace{\sum_{\chi} \cost{\psi}{\chi} \ell_{\rm coa}(\chi \mid \sigma)}} 
	\end{align}
	We think of $C_{\rm sp} (\psi \mid \sigma)$ as the expected cost per locus in the absence of gene tree heterogeneity. On the other hand, $C_{\rm coa} (\psi \mid \sigma)$ represents an additional cost that arises due to ILS and gene tree discordance. 
	
	$C_{\rm coa} (\psi \mid \sigma)$ is easy to compute in the limit where $\sigma$ is a star tree (i.e. all internal branch lengths of $\sigma$ are $0$). We denote a star tree by $\star$. Note once again that the external branch lengths of the star tree are irrelevant, since no coalescent events occur along these branches. For a star tree, the coalescence term in fact comprises the entirety of the expected cost. Under our modeling assumptions, the coalescent process in the singular ancestral population (where all internal branches of the gene trees arise) is that of the Kingman coalescent, and so the expected length of internal branches of the gene tree subtending exactly $2 \leq i \leq n-1$ tips is $2/i$ \citep{Fu1995}. We then have
	\begin{align*}
		\ell_{\rm coa}(\chi \mid \star) &= \frac{2}{|\chi|} \binom{n}{|\chi|}^{-1} \\ 
		C_{\rm coa} (\psi \mid \star) &= \mathlarger{\sum}_{\chi} \, \cost{\psi}{\chi} \, \frac{2}{|\chi|}\binom{n}{|\chi|}^{-1} 
	\end{align*}
	where the expression for  $\ell_{\rm coa}(\chi \mid \star)$  follows by symmetry, as there are $\binom{n}{|\chi|}$ possible branches subtending exactly $|\chi|$ tips. This provides a simple criterion to show inconsistency: if there exist two topologies $\psi, \psi' \in \T_n$ such that 
	\begin{align*}
		C(\psi \mid \star) < C(\psi' \mid \star) 
	\end{align*}
	then, owing to the continuity of $C(\psi \mid \sigma)$ as a function of $\bm{x}$, the concatenated counting estimator will always prefer $\psi$ over $\psi'$ for sufficiently short branch lengths $\bm{x}$, even if the true species topology is $\psi_*=\psi'$. 
	
	\subsection{A detailed analysis for the 4-taxa case}
	
	\noindent To demonstrate the procedure of finding $\ell(\chi \mid \sigma)$ in general, we work out two examples in the 4-taxa case. For these examples, we will take $\sigma$ to have the asymmetric topology $\psi_* = (((AB)C)D)$ and internal branch lengths $x_1, x_2$ subtending taxa $A,B$ and taxa $A,B,C$, respectively. Our goal will be to find the expected branch lengths $\ell(AB \mid \sigma)$ (Example \ref{ex:len_ab}) and $\ell(CD \mid \sigma)$ (Example \ref{ex:len_cd}). We will see that we can reuse much of the work from Example \ref{ex:len_ab} in Example \ref{ex:len_cd}, so we will provide most of the detailed exposition of the procedure in Example \ref{ex:len_ab}. We will then finish by observing that the sum $\ell(AB \mid \sigma) + \ell(CD \mid \sigma)$ has a particularly nice form, and use it to motivate a useful lemma (Lemma \ref{lemma:quartet_isps}), originally shown by \cite{molloy2022theoretical}.    
	\begin{example}\label{ex:len_ab}
		\rm 
		To compute $\ell(AB \mid \sigma)$, we start with the decomposition into species and coalescence terms as in \eqref{eq:len_decomp}: we know that $\ell_{\rm sp}(AB \mid \sigma) = x_1$, because the length of the internal branch of the species tree that subtends exactly $AB$ is $x_1$. Meanwhile, we can use \eqref{eq:inversion2} to find that
		\begin{align*}
			\ell_{\rm coa}(AB \mid \sigma) = - h_{\rm coa}(AB \mid \sigma) + h_{\rm coa}(ABC \mid \sigma) + h_{\rm coa}(ABD \mid \sigma) - h_{\rm coa}(ABCD \mid \sigma)
		\end{align*}
		\begin{figure}[h]
			\centering
			\includegraphics[width=15cm]{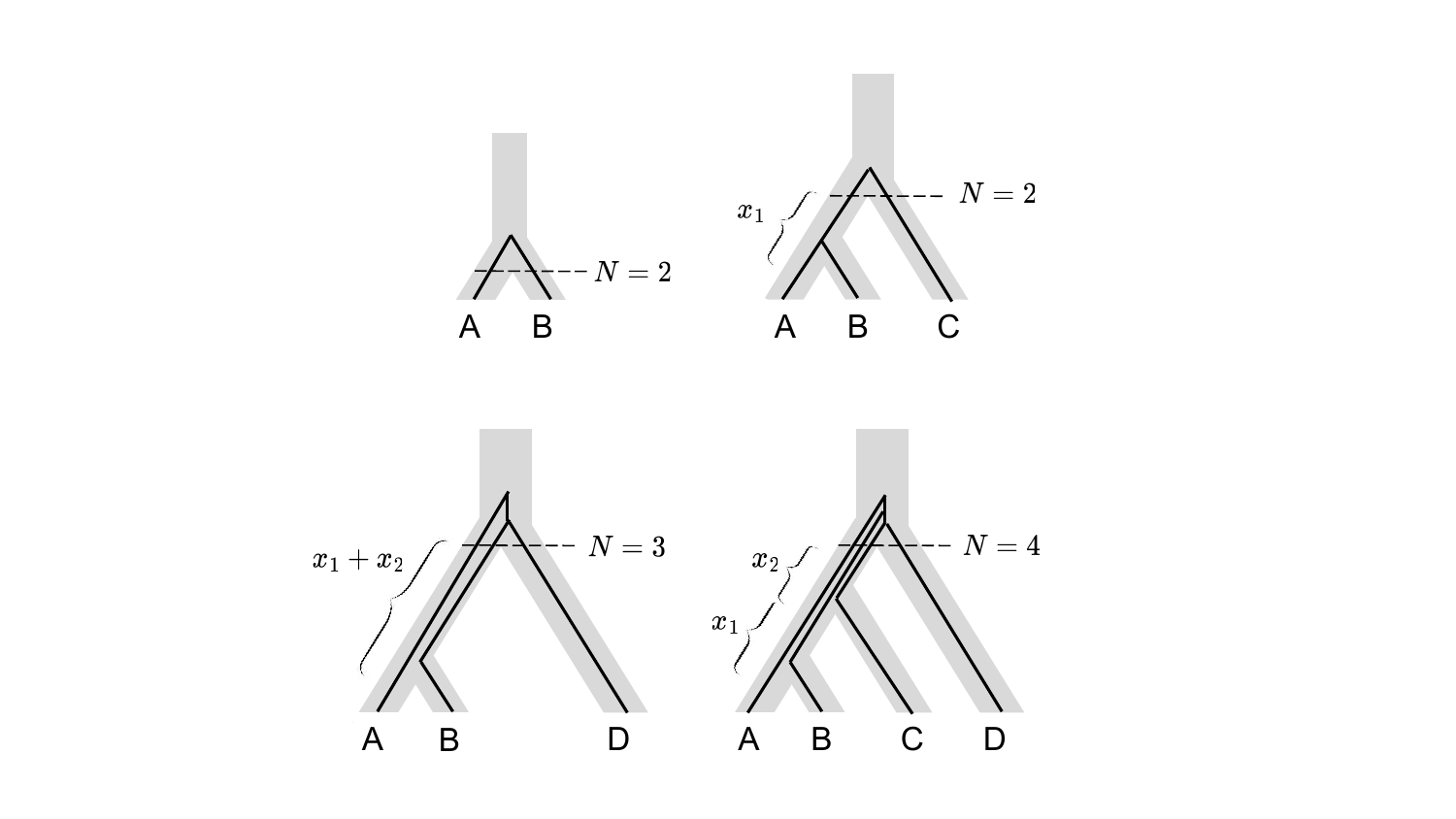}
			\caption{The restriction of a single species tree $\sigma$ with asymmetric topology $(((AB)C)D)$ to taxa $AB$, $ABC$, $ABD$, and $ABCD$. Example realizations of a gene tree for each restricted species tree and the random variable $N$ (the number of lineages entering the root of the restricted species tree) are shown  }\label{fig:rest_trees}
		\end{figure}
		\noindent Our next step will be to construct four sub-species trees of $\sigma$, obtained by  restricting $\sigma$ to the taxa $AB$, $ABC$, $ABD$ and $ABCD$, respectively. We will denote the species tree obtained by restricting to the taxa in $\chi$ by $\sigma_{\mid \chi}$. For each of these restricted trees, we must find the distribution of $N$, the number of lineages of the gene tree that enter the root of the restricted species tree. See Figure \ref{fig:rest_trees} for an illustration. To do so, we use the coalescent transition probability functions $g_{ij}(t)$ from \citep[Equation 6.1]{Tavare1984}. In this example, we will only need the specific formulae for $i=2,3$, listed below for reference.
		\begin{align*}
			g_{21}(t) &= 1-\exp(-t), \quad g_{22}(t) = \exp(-t) \\
			g_{31}(t) &= 1-\frac{3}{2}\exp(-t)+\frac{1}{2}\exp(-3t), \quad g_{32}(t) = \frac{3}{2}\exp(-t)-\frac{3}{2}\exp(-3t), \quad g_{33}(t) = \exp(-3t) 
		\end{align*}
		Once the distribution of $N$ is found, we may apply \cite[Equation 3]{efromovich2008coalescent} to find the expected height above the root of the restricted tree:
		\begin{equation*}
			h_{\rm coa}(\chi \mid \sigma) = 2-\sum_{i=2}^{|\chi|}\frac{2}{i} \, \P(N=i \mid \sigma_{\mid \chi})
		\end{equation*}
		We now carry out the calculation for each of the restricted species trees:
		\begin{itemize}
			\item For restricted species tree $\sigma_{\mid AB}$: Clearly,  $\P(N=1 \mid \sigma_{\mid AB})=1$, so $h_{\rm coa}(AB \mid \sigma)=1$.
			\item For restricted species tree $\sigma_{\mid ABC}$: $\P(N=2 \mid \sigma_{\mid ABC})=g_{21}(x_1) = 1-\exp(-x_1)$ (i.e. if lineages $A,B$ coalesce in the ancestral population of species $A,B$) and $P(N=3 \mid \sigma_{\mid ABC})=g_{22}(x_1)=\exp(-x_1)$. Thus, $h_{\rm coa}(ABC \mid \sigma)=2-\frac{2}{2}(1-\exp(-x_1)) - \frac{2}{3}\exp(-x_1) = 1+\frac{1}{3}\exp(-x_1)$.
			\item For restricted tree $\sigma_{\mid ABD}$: the same logic as above (with the only difference that $\sigma_{\mid ABD}$ has a single internal branch of length $x_1+x_2$) shows that $h_{\rm coa}(ABD \mid \sigma)=1+\frac{1}{3}\exp(-x_1-x_2)$.
			\item For restricted species tree $\sigma_{\mid ABCD}$: we first must consider the number of lineages entering and exiting the ancestral population of taxa $A,B,C$. Two lineages of the gene tree enter ancestral population of taxa $A,B,C$ with probability $g_{21}(x_1)=1-\exp(-x_1)$ and three lineages of the gene tree enter the ancestral population with probability $g_{22}(x_1)=\exp(-x_1)$. Therefore, by considering the number of coalescence events that occur, we have either 1, 2 or 3 lineages of the gene tree at the top of the ancestral population of taxa A,B,C. This leads rise to 2, 3, or 4 lineages that enter the root population when also considering lineage D:
			\begin{align*}
				\P(N=2 \mid \sigma) &= g_{31}(x_2) g_{22}(x_1) + g_{21}(x_2) g_{21}(x_1) = 1 - \exp(-x_2) - \frac{1}{2}\exp(-x_1-x_2) + \frac{1}{2}\exp(-x_1-3x_2) \\ 
				\P(N=3 \mid \sigma) &= g_{32}(x_2) g_{22}(x_1) + g_{22}(x_2) g_{21}(x_1) = \exp(-x_2) + \frac{1}{2}\exp(-x_1-x_2) - \frac{3}{2}\exp(-x_1-3x_2)\\
				\P(N=4 \mid \sigma) &= g_{33}(x_2) g_{22}(x_1) = \exp(-x_1-3x_2)
			\end{align*}
			Putting this together gives $h_{\rm coa}(ABCD \mid \sigma) = 1+\frac{1}{3}\exp(-x_2)+\frac{1}{6}\exp(-x_1-x_2)$.
		\end{itemize}
		With these expected heights, we find that 
		\begin{align*}
			\ell_{\rm coa}(AB \mid \sigma) &= - 1 + \big[1+\frac{1}{3}\exp(-x_1)\big] +\big[1+\frac{1}{3}\exp(-x_1-x_2)\big] - \big[1+\frac{1}{3}\exp(-x_2)+\frac{1}{6}\exp(-x_1-x_2)\big]\\
			&=\frac{1}{3}\exp(-x_1) - \frac{1}{3}\exp(-x_2) + \frac{1}{6}\exp(-x_1-x_2)
		\end{align*}
		and adding back the species term gives 
		\begin{equation*}
			\ell(AB \mid \sigma) = x_1 + \frac{1}{3}\exp(-x_1) - \frac{1}{3}\exp(-x_2) + \frac{1}{6}\exp(-x_1-x_2)
		\end{equation*}
		or, using the transformed variables $(X_1, X_2) = (\exp(-x_1), \exp(-x_2))$ as shorthand, 
		\begin{equation*}
			\ell(AB \mid \sigma) = x_1 + \frac{1}{3}(X_1-X_2) + \frac{1}{6}X_1X_2
		\end{equation*}
		
	\end{example}
	
	\begin{example}\label{ex:len_cd}
		\rm
		To compute $\ell(CD \mid \sigma)$, we again decompose into species and coalescence terms, noting $\ell_{\rm sp}(CD \mid \sigma) = 0$ since no branch of the species tree subtends exactly the taxa $C,D$. Therefore, the entire expected length $\ell(CD \mid \sigma)$ consists of the coalescence term:   
		\begin{align*}
			\ell(CD \mid \sigma) = - h_{\rm coa}(CD \mid \sigma) + h_{\rm coa}(ACD \mid \sigma) + h_{\rm coa}(BCD \mid \sigma) - h_{\rm coa}(ABCD \mid \sigma)
		\end{align*}
		Here, we can reuse many of the calculations from Example 1, both in a direct and indirect manner. For instance, we may reuse our calculation for $h_{\rm coa}(ABCD \mid \sigma)$ directly, while applying the same logic as in our calculations for $h_{\rm coa}(ABC \mid \sigma)$ to determine $h_{\rm coa}(ACD \mid \sigma)$ and $h_{\rm coa}(BCD \mid \sigma)$. In particular, $\sigma_{\mid ABC}$ is a 3-taxa caterpillar tree with the singular internal branch length $x_1$, whereas  $\sigma_{\mid ACD}$ and $\sigma_{\mid BCD}$ are both 3-taxa caterpillar trees with internal branch length $x_2$. So $h_{\rm coa}(ACD \mid \sigma)$ can be found by substituting $x_2$ in for $x_1$ in the expression for $h_{\rm coa}(ABC \mid \sigma)$. Putting these observations together, we get 
		\begin{align*}
			\ell(CD\mid \sigma) &= - 1 + \big[1+\frac{1}{3}\exp(-x_2)\big] +\big[1+\frac{1}{3}\exp(-x_2)\big] - \big[1+\frac{1}{3}\exp(-x_2)+\frac{1}{6}\exp(-x_1-x_2)\big]\\
			&=\frac{1}{3}\exp(-x_2)- \frac{1}{6}\exp(-x_1-x_2)
		\end{align*}
		or, once again using the transformed variables $(X_1, X_2) = (\exp(-x_1), \exp(-x_2))$ as shorthand, 
		\begin{equation*}
			\ell(CD \mid \sigma) = \frac{1}{3}X_2 - \frac{1}{6}X_1X_2
		\end{equation*}
		
	\end{example}
	
	Using the results of Examples 1 and 2, we can make the observation that $\ell(AB \mid \sigma) + \ell(CD \mid \sigma)$, which is proportional to the expected number of informative site patterns that support the unrooted quartet $AB|CD$, is in fact independent of $x_2$:
	\begin{equation} \label{eq:quartet_isps}
		\ell(AB \mid \sigma) + \ell(CD \mid \sigma) = x_1 + \frac{1}{3}X_1
	\end{equation}
	This result can be thought of in terms of the unrooted species tree $\overline{\sigma}$ obtained by unrooting $\sigma$. Indeed, $\overline{\sigma}$ has a single internal branch of length $x_1$, with the previous internal branch length of $x_2$ in the rooted tree $\sigma$ being collapsed into an external branch of $\overline{\sigma}$.  Moreover, we can use this result to quickly compute the expected number of  number of informative site patterns that support the alternative quartets $AC|BD$ or $AD|BC$. To see how, we note that for any permutation $a,b,c,d$ of $A,B,C,D$, we have 
	\begin{align*}
		\ell(ab \mid \sigma)+\ell(cd \mid \sigma) = -&h(ab \mid \sigma) +h(abc \mid \sigma) +h(abd \mid \sigma) -h(abcd \mid \sigma) \\
		-&h(cd \mid \sigma) +h(acd \mid \sigma) +h(bcd \mid \sigma) -h(abcd \mid \sigma)
	\end{align*}
	We observe that regardless of the choice of the permutation $a,b,c,d$ of $A,B,C,D$, the terms $+h(ABC)$, $+h(ABD)$, $+h(ACD)$, $+h(BC)$, and $-h(ABCD)$ (twice) will appear in this expansion in some order. That is, 
	\begin{equation}\label{eq:quartet_support}
		\ell(ab \mid \sigma)+\ell(cd \mid \sigma) = -h(ab \mid \sigma) -h(cd \mid \sigma) + \textrm{const}.
	\end{equation}
	where the constant is the same regardless of the choice of permutation $a,b,c,d$ of $A,B,C,D$, depending only on $\sigma$. Further, since the coalescent terms $h_{\rm coa}(ab \mid \sigma)=1$ regardless of the choice of $a,b$, we have that 
	\begin{equation*}
		\ell_{\rm coa}(ab \mid \sigma)+\ell_{\rm coa}(cd \mid \sigma) = \textrm{const.}
	\end{equation*}
	and we can read off this constant as $\frac{1}{3}X_1$ from \eqref{eq:quartet_isps}. Meanwhile the species terms are also $0$ for the alternative quartets $AC|BD$ and $AD|BD$: 
	\begin{align*}
		\ell_{\rm sp}(AC \mid \sigma)+\ell_{\rm sp}(BD \mid \sigma) = \ell_{\rm sp}(AD \mid \sigma)+\ell_{\rm sp}(BC \mid \sigma) = 0
	\end{align*}
	because no internal branches of $\sigma$ subtend exactly $AC$, $BD$, $AD$, or $BC$. Therefore, we conclude
	\begin{itemize}
		\item $\len{AB}{\sigma} + \len{CD}{\sigma} = x_1 + \frac{1}{3}X_1$; 
		\item $\len{AC}{\sigma} + \len{BD}{\sigma} = \len{AD}{\sigma} + \len{BC}{\sigma} =\frac{1}{3}X_1$ 
	\end{itemize}
	While we have only examined the case where $\sigma$ has an asymmetric topology, it turns out this result can be extended to the case where $\sigma$ has the symmetric topology $\psi_* = ((ab)(cd))$ as well, as in the following lemma, previously stated and proven in \cite{molloy2022theoretical} (though we rewrite it in our notation):
	\begin{lemma}\label{lemma:quartet_isps}
		Suppose $\sigma$ is a 4-taxa rooted species tree with taxa $A,B,C,D$, and let $a,b,c,d$ be a permutation of $A,B,C,D$. Then  
		\begin{itemize}
			\item $\len{ab}{\sigma} + \len{cd}{\sigma} = \tau + \frac{1}{3}\exp(-\tau)$ if $\overline{\sigma}$ displays the quartet $ab|cd$;
			\item $\len{ab}{\sigma} + \len{cd}{\sigma} = \frac{1}{3}\exp(-\tau)$ if $\overline{\sigma}$ displays does not display the quartet $ab|cd$;
		\end{itemize}
		where $\tau$ is the length of the singular internal branch of the unrooted species tree $\overline{\sigma}$ obtained by unrooting $\sigma$.
	\end{lemma}
	In words, this lemma tells us that in the 4-taxa case, 1) most informative site patterns support the quartet displayed by the unrooted species tree over the two alternative quartets; and 2) the magnitude of this support is an increasing function of the internal branch of the unrooted species tree. We will make use of this observation in Section \ref{SS:unrooted_5} when discussing the consistency of parsimony in the unrooted $5$-taxa case. 
	
	\section{Where parsimony succeeds and fails across tree space}
	
	\subsection{Parsimony for the unrooted 5-taxa case} \label{SS:unrooted_5}
	
	There is only one possible shape for an unrooted $5$-taxa topology -- all topologies are of form $((ab)(cd)e)$ (for some permutation $a,b,c,d,e$ of $A,B,C,D,E$), i.e. the two pairs of sister taxa $a,b$ and $c,d$ are separated from taxon $e$ (Figure \ref{fig:5taxa_unrooted_species}a). Therefore, we will examine one labeled representative of this shape only, assuming that species tree has true unrooted topology $\overline{\psi_*} = ((AB)(CD)E)$. The most direct method to analyze the consistency of concatenated parsimony (and a method can be generalized to all unrooted concatenated counting methods) would proceed as follows: 
	\begin{enumerate}
		\item Find all rooted topologies $\psi_*$ that have unrooted topology $((AB)(CD)E)$;
		\item Compute $C(\overline{\psi} \mid \sigma)$ as a function of $\bm{x}$ for $\sigma = (\psi_*, \bm{x}, \bm{y})$ for each $\overline{\psi} \in \overline{\T}_5$;
		\item Check if $C(\overline{\psi} \mid \sigma)$  is minimized at $\overline{\psi}_*=((AB)(CD)E)$ for all choices of nonzero branch lengths $\bm{x}$.
	\end{enumerate}
	We will ultimately follow a procedure analogous to this in the unrooted 6-taxa case (Section \ref{SS:unrooted_6}). However, such an approach would overlook the particularly simple structure of parsimony in the unrooted $5$-taxa case. For the candidate unrooted topology $\overline{\psi} = ((ab)(cd)e)$, parsimony assigns a cost of $1$ to the site patterns $ab, abe, cd, cde$ and a cost of $2$ to all other site patterns. As a result, as the number of loci sampled grows large, concatenated parsimony will prefer the topology $((ab)(cd)e)$ for which the sum of lengths 
	\begin{equation*}
		P(\overline{\psi} \mid \sigma) \defeq \len{ab}{\sigma} + \len{abe}{\sigma} + \len{cd}{\sigma} +  \len{cde}{\sigma}
	\end{equation*}
	is maximal. Noting the irrelevance of taxon $e$ in $P(\overline{\psi} \mid \sigma)$, it should not be difficult to believe that $P(\overline{\psi} \mid \sigma)$ depends only on the restricted species tree $\sigma_{\mid abcd}$:
	\begin{align*}
		P(\overline{\psi} \mid \sigma) = \len{ab}{\sigma_{\mid abcd}} + \len{cd}{\sigma_{\mid abcd}}
	\end{align*}
	For a more rigorous argument of this fact, one may expand each of the lengths appearing in $P(\overline{\psi} \mid \sigma)$ using \eqref{eq:inversion2}. For instance, when expanding $\len{ab}{\sigma}$ and $\len{abe}{\sigma}$, we get (after canceling terms)
	\begin{align*}
		\len{ab}{\sigma} + \len{abe}{\sigma} &= -h(ab \mid \sigma) + h(abc \mid \sigma) + h(abd \mid \sigma) - h(abcd \mid \sigma) \\ 
		&= -h(ab \mid \sigma_{\mid abcd}) + h(abc \mid \sigma_{\mid abcd}) + h(abd \mid \sigma_{\mid abcd}) - h(abcd \mid \sigma_{\mid abcd}) \\ 
		&= \len{ab}{\sigma_{\mid abcd}}  
	\end{align*}
	where the second line follows as $e$ does not appear in any term $\pm h(\chi \mid \sigma)$ in the first line (hence, taxon $e$ may be omitted from the species tree). A similar argument shows $\len{cd}{\sigma} + \len{cde}{\sigma} = \len{cd}{\sigma_{\mid abcd}}$. 
	
	Since we are now dealing only with the 4-taxa tree $\sigma_{\mid abcd}$, we are in a position to apply Lemma \ref{lemma:quartet_isps}. In particular,
	$P(\overline{\psi} \mid \sigma)$ is indeed maximized at $\overline{\psi_*} = ((AB)(CD)E)$, because 1) $\sigma_{\mid ABCD}$ displays the quartet $AB|CD$, and 2) the length of the internal branch of $\overline{\sigma}_{\mid abcd}$ is maximal when $\{a,b,c,d\} = \{A,B,C,D\}$. (See Figure \ref{fig:5taxa_unrooted_species} for an illustration.) We therefore conclude that concatenated parsimony is statistically consistent in the unrooted $5$-taxa case under the MSC + infinite-sites model of evolution as presented in Section \ref{S:def_result}. 
	
	\begin{figure}[h]
		\centering
		\includegraphics[width=15cm]{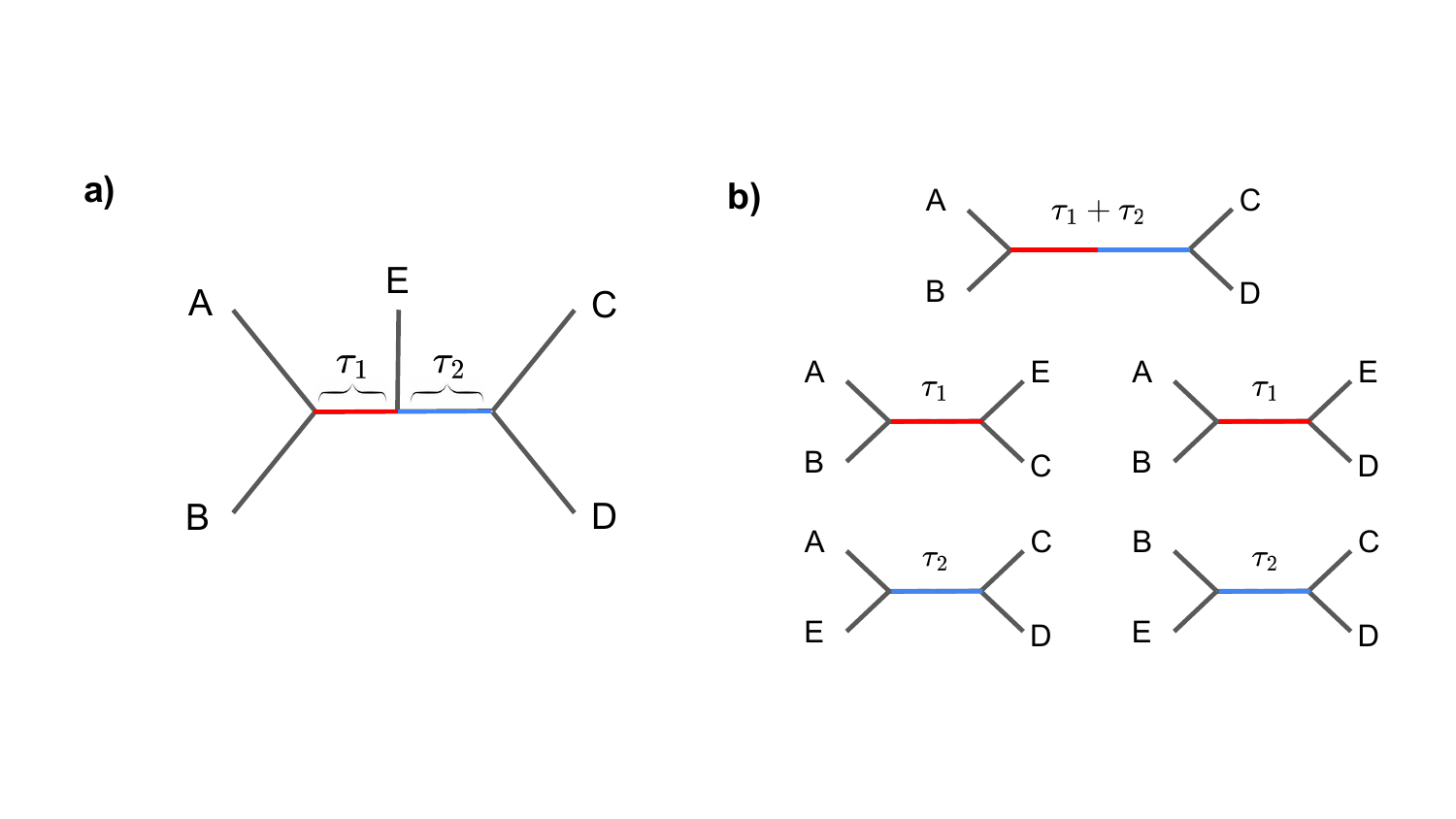}
		\caption{a) An unrooted 5-taxa species tree $\overline{\sigma}$ with unrooted topology $\overline{\psi}_*=((AB)(CD)E)$ and internal branch lengths $\tau_1, \tau_2$. b) The five quartet trees displayed by the unrooted species tree $\overline{\sigma}$, showing that $\overline{\sigma}_{\mid ABCD}$ has the longest internal branch among them. } \label{fig:5taxa_unrooted_species}
	\end{figure}
	
	We stress that if the assumptions of the MSC + infinite-sites model of evolution are not met, then we cannot use the above result to make a conclusion about the consistency of parsimony in this case. However, we expect that many of the strict assumptions that we made on the particular details of the infinite-sites model of mutation are not necessary for consistency. For example, in our proof of consistency we have used the fact that as the number of loci sampled grows large, most informative site patterns support the quartet displayed by $\sigma_{\mid abcd}$ over the two alternative quartets. Let us generalize this case to the scenario where the scaled mutation rate $\theta = 4N_e \mu$ is allowed to vary within and across species tree branches (though we will assume it remains bounded between two positive numbers). We could also allow variability across loci, but this requires some additional theoretical care. We consider the expected contribution of the $i^{\rm th}$ locus to the overall cost in \eqref{eq:tot_cost_unrooted}; we want to show that this expected contribution is minimal when the candidate topology is taken to be the quartet displayed by $\overline{\sigma}_{abcd}$ as compared to the two alternative quartets. 
	
	We begin by rescaling all time to be in mutation units, such that mutations fall on the gene tree $G_i$ at constant rate $1$. Note that the species tree may no longer be an ultrametric tree with respect to the branch lengths given in mutation units, which complicates our notion of height used throughout the paper. However, we can easily resolve this issue by appending additional length to the external branches of the species tree so that it becomes ultrametric again -- this is possible since no informative site patterns result from mutations occurring in these external branches. The rescaling of time also causes the pairwise coalescent rate to vary inversely to scaled mutation rate across the species tree, apparently further complicating our analysis. However, when time is measured in mutation units, sites in the alignment $A_i$ for locus $i$ that support a particular quartet $ab|cd$ within a $4$-taxa tree $\sigma_{\mid abcd}$ is still proportional (or in this case, equal) to $\ell^{\theta}(ab \mid \sigma_{\mid abcd}) + \ell^{\theta}(cd \mid \sigma_{\mid abcd})$, where we use the ``$\theta$`` superscript to denote that these expected lengths are in mutation units. We then can apply \eqref{eq:quartet_support}, which tells us that $\ell^{\theta}(ab \mid \sigma_{\mid abcd}) + \ell^{\theta}(cd \mid \sigma_{\mid abcd})$ is maximal when $h^{\theta}(ab \mid \sigma_{\mid abcd}) + h^{\theta}(cd \mid \sigma_{\mid abcd})$ is minimal. Using this observation, it is easy to verify (without any explicit calculation) that most informative site patterns will still support the quartet displayed by $\sigma_{\mid abcd}$ over the two alternative quartets. For example, if $\sigma_{\mid abcd}$ has an asymmetric topology $(((ab)c)d)$, then 
	\begin{align*}
		h^{\theta}(ab \mid \sigma_{\mid abcd}) &< h^{\theta}(bc \mid \sigma_{\mid abcd}) = h^{\theta}(ac \mid \sigma_{\mid abcd}) \\ 
		h^{\theta}(cd \mid \sigma_{\mid abcd}) &= h^{\theta}(ad \mid \sigma_{\mid abcd}) = h^{\theta}(bd \mid \sigma_{\mid abcd})
	\end{align*}
	which shows that $h^{\theta}(ab\mid\sigma_{\mid abcd}) + h^{\theta}(cd\mid \sigma_{\mid abcd})$ is less than $ h^{\theta}(ac\mid\sigma_{abcd}) + h^{\theta}(bd \mid\sigma_{\mid abcd}) $ and  $ h^{\theta}(ad\mid\sigma_{\mid abcd}) + h^{\theta}(bc \mid \sigma_{\mid abcd})$, as desired.  
	
	\subsection{The parsimony anomaly zone for the rooted 5-taxa case}  \label{SS:rooted_5}
	
	There are three possible shapes for rooted, 5-taxa topologies (i.e. three possible topologies up to relabeling). To analyze the consistency/inconsistency of parsimony across all of parameter space, it suffices to choose one labeled representative of each such shape. Accordingly, we let $\psi_1 = ((((AB)C)D)E), \psi_2 = (((AB)C)(DE)), \psi_3 = (((AB)(CD))E)$ be labeled representatives of these three shapes for 5-taxa, and define $[\psi_i]$ for $i=1,2,3$ to be the collection of all rooted topologies which agrees in unlabeled topology with $\psi_i$, i.e. $\psi \in [\psi_i]$ if and only if there is a permutation $\pi$ of $\{A,B,C,D,E\}$ such that $\pi(\psi) = \psi_i$. We then consider three species trees $\sigma_{5,i} \defeq (\psi_i, \bm{x})$ i.e. $\sigma_{5,i}$ is a species tree with topology $\psi_i$ and the particular branch length assignment $\bm{x}$ demonstrated in Figure \ref{fig:5taxa_labeled}.
	\begin{figure}[h]
		\centering
		\sbox0{\includegraphics{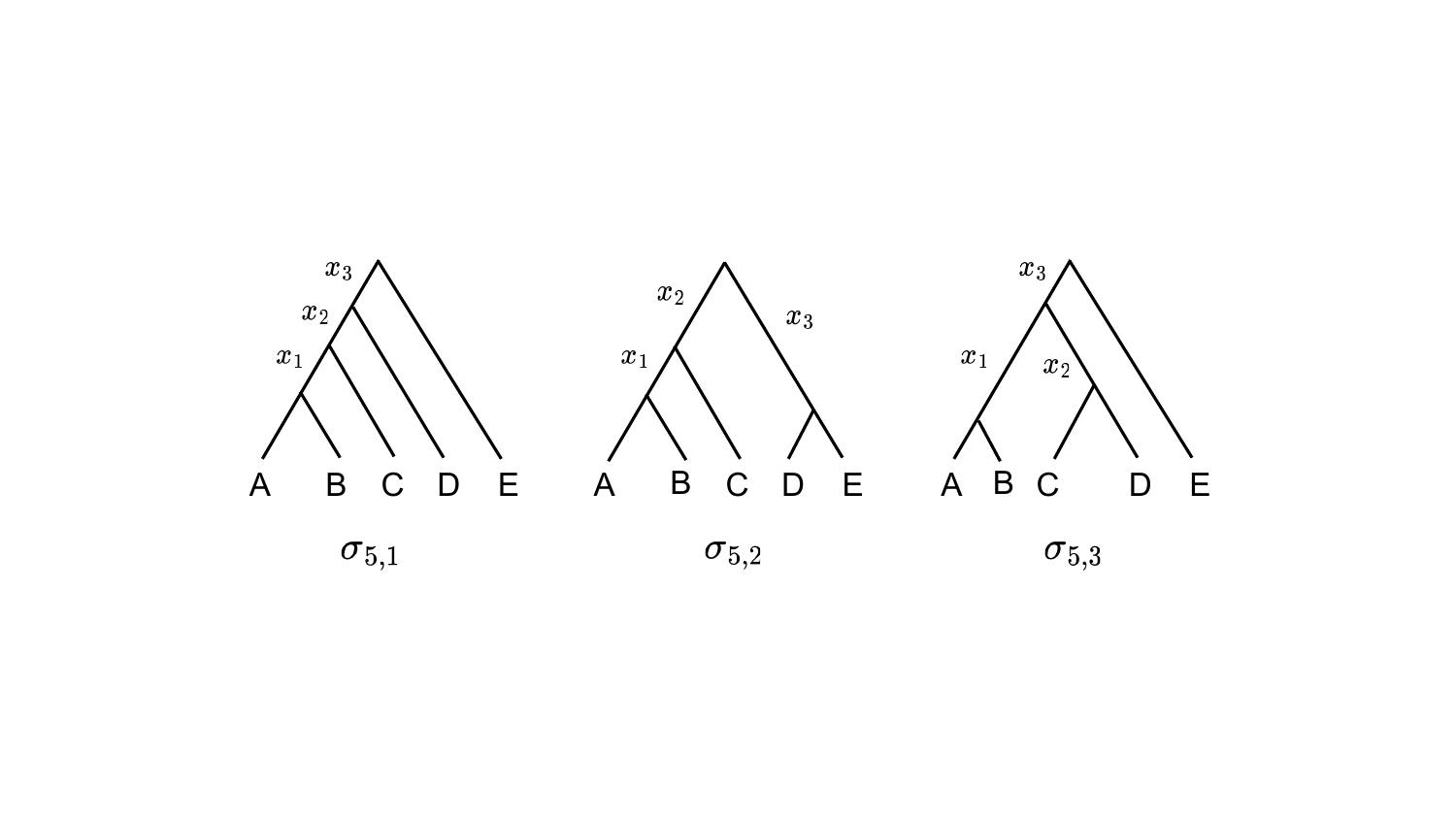}}
		\includegraphics[trim={0.05\wd0} {0.15\wd0} {0.05\wd0} {0.15\wd0}, clip, width=16cm]{Figure7.pdf}
		\caption{The species trees $\sigma_{5,1}, \sigma_{5,2}, \sigma_{5,3}$, with respective topologies $\psi_1, \psi_2, \psi_3$ and branch lengths $\bm{x} = (x_1, x_2, x_3)$. External branch lengths are irrelevant for analysis and are thus omitted. }\label{fig:5taxa_labeled}
	\end{figure}
	
	To begin, we can show that concatenated parsimony fails to be statistically consistent when the species tree has sufficiently short internal branch lengths. Applying the decomposition introduced earlier under a star tree (Section \ref{SS:st}), $\psi_3$ has a lower average cost per locus than $\psi_1$ and $\psi_2$: 
	\begin{equation*}
		\gencost{C}{\psi_3}{\star}  = \frac{43}{10} < \frac{13}{3} = \gencost{C}{\psi_1}{\star}  = \gencost{C}{\psi_2}{\star} 
	\end{equation*}
	Owing to the symmetry of the star tree, the choice of labeled representatives $\psi_1, \psi_2, \psi_3$ from $[\psi_1], [\psi_2], [\psi_3]$ is clearly irrelevant in the above relations. Thus, all $15$ candidate topologies with $\psi \in [\psi_3]$ are preferred over a true species tree topology $\psi_* \in [\psi_1] \cup [\psi_2]$ for sufficiently short internal branch lengths $\bm{x}$. However, this argument does not characterize the exact regions where concatenated parsimony will fail to be consistent. 
	
	To visualize the region of inconsistency, we define $K(\sigma)$ to be the number of candidate topologies preferred over the true species tree topology $\psi_*$:
	\begin{equation*}
		K(\sigma) := \sum_{\psi \neq \psi_* \in \T_5} 1\{\gencost{C}{\psi}{\sigma} \leq \gencost{C}{\psi_*}{\sigma}\}
	\end{equation*}
	We call a candidate topology $\psi$ that contributes $+1$ to $K(\sigma)$ as a parsimony anomalous gene tree (PAGT) for $\sigma$. The values of $K(\sigma)$, considered for a fixed species tree $\psi_*$ as a function of $\bm{x}$, form what we will call the parsimony anomaly zone for the topology $\psi_*$. To visualize the parsimony anomaly zone for the three representative species tree topologies $\psi_1, \psi_2, \psi_3$, we computed $K(\sigma_{5,i})$ as a function of $\bm{x}$ for $i=1,2,3$, varying $x_3 \in \{0,1/50,1/10,1/2\}$ and varying $x_1, x_3$ each across 400 uniformly spaced values in $[0,0.1]$. We interpolated between grid points to generate filled contours using the \textit{contourf()} function in \textit{matplotlib} \citep{Hunter:2007}. The results are given in Figure \ref{fig:PAGTS_5taxa}. 
	\begin{figure}[h]
		\centering 
		\sbox0{\includegraphics{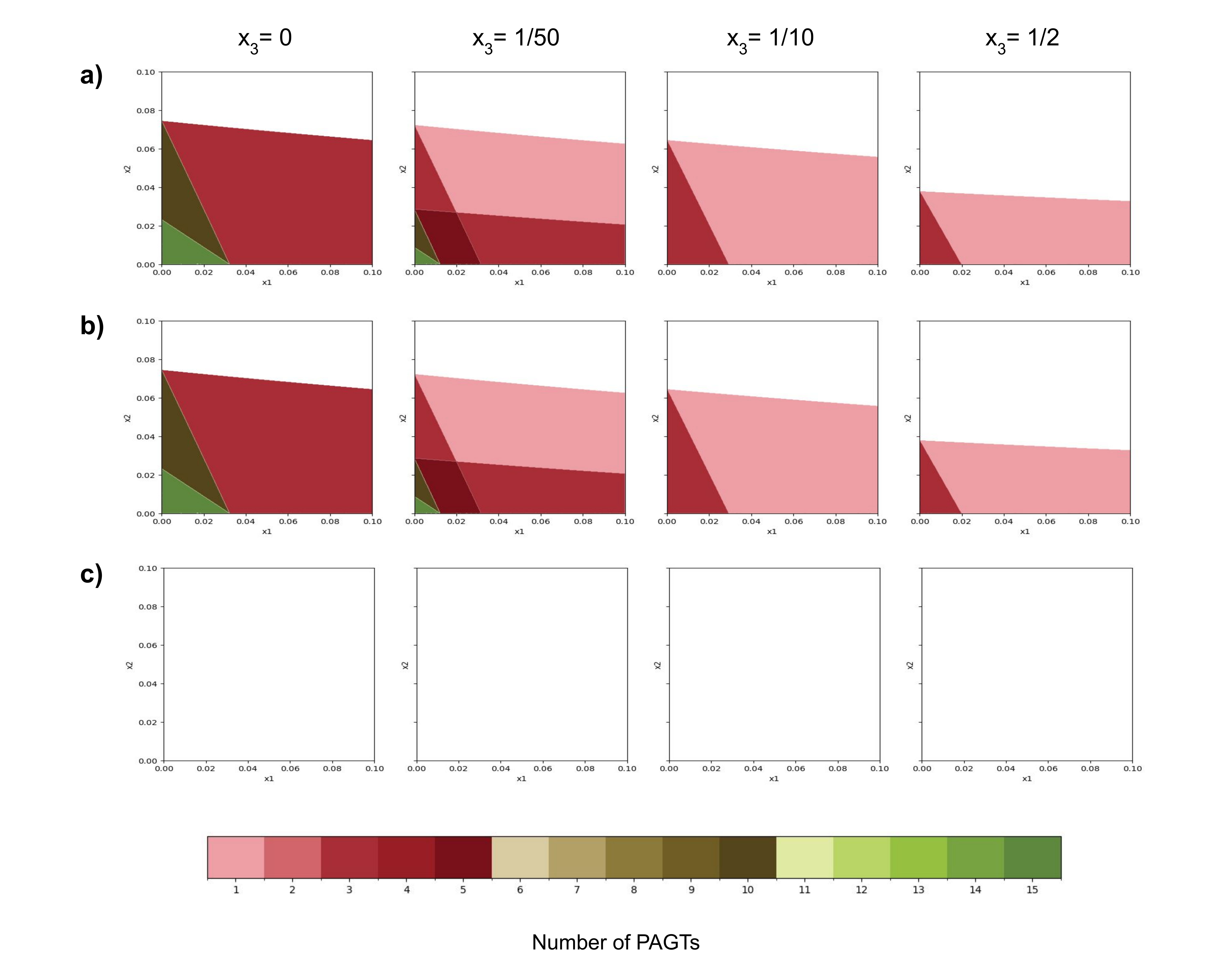}}
		\includegraphics[trim={.05\wd0} {.01\wd0} {.05\wd0} {0\wd0}, clip, width=15cm]{Figure8.pdf}
		\caption{Number of parsimony anomalous gene topologies preferred by parsimony over the true species tree topology for a) $\psi_1$, b) $\psi_2$, c) $\psi_3$, visualized a function of the internal branch lengths $\bm{x}$ of the corresponding species trees $\sigma_{5,1}$, $\sigma_{5,2}$ and $\sigma_{5,3}$ . Each column has a fixed value of $x_3$ (given in coalescent units), with $x_1$ and $x_2$ variable across $[0, 0.1]$ coalescent units in each plot.}\label{fig:PAGTS_5taxa}
	\end{figure}
	
	From Figure \ref{fig:PAGTS_5taxa}, we see that when the true species tree topology $\psi_*$ has an unlabeled topology agreeing with $\psi_3$, parsimony always prefers the true topology $\psi_*$, regardless of the branch lengths chosen. Unfortunately, researchers do not know when this topology is the true one, and so this fact on its own is not useful for applications. Meanwhile, when the species tree topology $\psi_* \in [\psi_1]$ or $\psi_* \in [\psi_2]$, there are regions in which parsimony prefers other topologies over the true species tree topology, i.e. regions in which parsimony will be statistically inconsistent. 
	
	By coincidence, the species tree topologies $\psi_1$ and $\psi_2$ appear to share an identically structured parsimony anomaly zone with the choice of branch lengths assignment for $\sigma_{5,1}$ and $\sigma_{5,2}$, as in Figure \ref{fig:5taxa_labeled}. However, the particular topologies that are anomalous may differ by region. In either case, we can observe there are exactly $15$ PAGTs near $\bm{x} = \bm{0}$, exactly as suggested by the star tree argument. It is then interesting to ask which (if any) of the $15$ topologies in $[\psi_3]$ are maximally anomalous for $\psi_i$ $(i=1,2)$, in the sense that they are always preferred over any other labeled variant in $[\psi_3]$ regardless of the choice of branch lengths $\bm{x}$ in $\sigma_{5,i}$. It turns out that this may be answered relatively easily, due to the expected cost per locus $C(\psi \mid \sigma_{5,i})$ for $i=1,2$ taking a surprisingly similar form for all $\psi \in \psi_3$. Recall the decomposition of the expected cost per locus into a 'species' and 'coalescence' term as in Section \ref{SS:st}: 
	\begin{equation*}
		C_{\rm sp} (\psi \mid \sigma) + C_{\rm coa} (\psi \mid \sigma) 
	\end{equation*}
	We have verified (see the Jupyter notebook code in \href{S:data}{Data availability}) that for a species tree topology $\psi_* \in [\psi_1] \cup [\psi_2]$, the coalescent term $C_{\rm coa} (\psi \mid \sigma_{5,i})$ is identical for all $\psi \in [\psi_3]$ for fixed $i=1,2$, and this term even agrees between the two species trees $\sigma_{5,1}, \sigma_{5,2}$ (which helps explain the identically shaped parsimony anomaly zones). In particular, letting $(X_1,X_2,X_3) = (\exp(-x_1) ,\exp(-x_2), \exp(-x_3))$, we have that for all $\psi \in [\psi_3]$,
	\begin{align*}
		& C_{\rm coa} (\psi \mid \sigma_{5,1}) = C_{\rm coa} (\psi \mid \sigma_{5,2}) \\ =& \,X_1 + X_2 + X_3 + \frac{1}{2}X_1X_2 + \frac{1}{2}X_2X_3 + \frac{1}{4}X_1X_2X_3 + \frac{1}{20}X_1X_2^3X_3
	\end{align*}
	In words, this result tells us that on average, ILS and gene tree discordance causes an equal additional parsimony cost to each $\psi \in [\psi_3]$ for a species tree with topology $\psi_* \in  [\psi_1] \cup [\psi_2]$. Therefore, the maximally anomalous topology $\psi \in [\psi_3]$ for $\psi_*=\psi_i$ ($i=1,2$) will be the one that minimizes the species term $C_{\rm sp} (\psi \mid \sigma_{5,i})$, i.e. the topology in $[\psi_3]$ that would be the most parsimonious in the absence of gene tree heterogeneity. For $\sigma_{5,1}$, this implies that the maximally anomalous topology is $\psi_3$, with $C_{\rm sp} (\psi \mid \sigma_{5,1}) = x_1+2x_2+x_3$. This expression follows since the site patterns $AB, ABCD$ both have a parsimony cost of $1$ on $\psi_3$, while the site pattern $ABC$ has a parsimony cost of 2 on $\psi_3$. Meanwhile, for the species tree $\sigma_{5,2}$, the topology $\widetilde{\psi}_3 = (((AB)(DE))C)$ is maximally anomalous.   
	
	Using the fact that $\psi_3$ is maximally anomalous for $\psi_1$, finding the exact shape of the parsimony anomaly zone may be found by solving $\gencost{C}{\psi_1}{\sigma_{5,1}} \geq \gencost{C}{\psi_3}{\sigma_{5,1}}$, obtaining an inequality on $x_1$ in terms of $x_2, x_3$: 
	\begin{equation}\label{eq:anom_boundary}
		x_1 \leq \log \left[ \frac{1-X_2+\frac{X_2^3X_3}{10}}{3x_2-X_3+X_3X_2} \right]
	\end{equation}
	Similar bounds may be found that define the exact regions in which other candidate topologies $\psi \in [\psi_3]$ are anomalous.  We can also partially visualize the overall geometry of the anomaly zone by calculating $K(\sigma_{5,i})$ for $i=1,2$ in the degenerate cases when $x_1=0$ or $x_2=0$. In particular, we varied the other two internal branch lengths $x_j, x_k$ across $[0,0.1]$ coalescent units, sampling $400$ uniformly spaced values for each $x_j, x_k$. The results are shown in Figure \ref{fig:PAGTS_5taxa_3d}. We report the results for $\sigma_{5,1}$ only, since we once again observed in our data that both $\sigma_{5,1}$ and $\psi_{5,2}$ gave an identically shaped parsimony anomaly zone.  
	\begin{figure}[ht]
		\centering 
		\includegraphics[width=15cm]{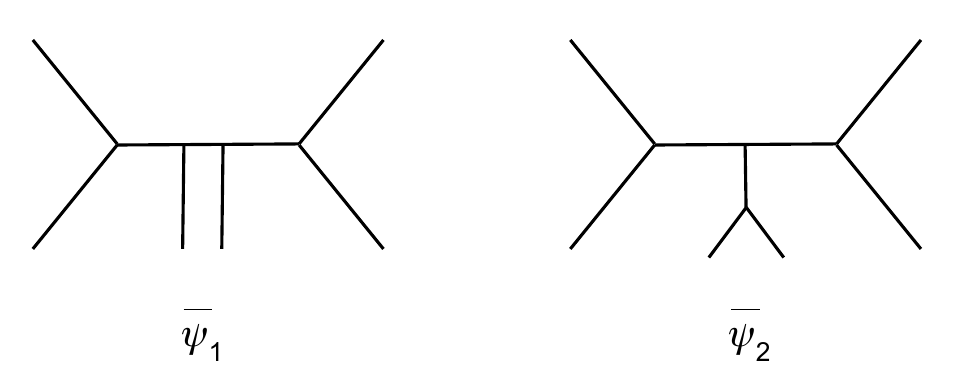}
		\caption{Number of topologies preferred by parsimony over the true species tree topology for the species tree $\sigma_{5,1}$ as shown in Figure \ref{fig:5taxa_labeled}. Each of the three surfaces shown corresponds to the cases $x_i=0$ for $i=1,2,3$, with $x_j, x_k$ for $j\neq k \neq i$ variable across $[0,0.1]$ coalescent units.}\label{fig:PAGTS_5taxa_3d}
	\end{figure}
	Of particular interest in Figure \ref{fig:PAGTS_5taxa_3d} is that when $x_2=0$, there is always at least one PAGT (namely, the maximally anomalous topology $\psi_3$) regardless of how large $x_1$ and $x_3$ are. This demonstrates that it is in general necessary for all branch lengths $x_i$ to exceed some critical threshold $x_{\rm min}$ in order to guarantee the consistency of parsimony. To find this threshold, assume we have a species tree $\sigma = (\psi_1, \bm{x})$ with all internal branches having the same length $x_i = x$ for some $x > 0$. Making this substitution in \eqref{eq:anom_boundary} and setting both sides equal, we find a solution of $x_{\rm min} \approx 0.062205$ coalescent units. In comparison, for a species tree with topology $\psi_1$, \cite{rosenberg2008discordance} showed that the critical threshold of minimum branch length for the most probable gene tree topology to agree with the species tree topology was $x_{\rm \min} \approx 0.1935$ coalescent units, so parsimony still outperforms the democratic vote method by a non-trivial margin in the rooted 5-taxa case.
	
	\subsection{The parsimony anomaly zone for the unrooted 6-taxa case} \label{SS:unrooted_6}
	
	When computing the number of PAGTs in the unrooted 6-taxa case, the definition of the number of PAGTs $K(\sigma)$, must be updated to use unrooted topologies of $\overline{\T}_6$:
	\begin{equation*}
		K(\sigma) = \sum_{\overline{\psi} \neq \overline{\psi_*} \in \overline{\T}_6} 1\{\gencost{C}{\overline{\psi}}{\sigma} \leq \gencost{C}{\overline{\psi}_*}{\sigma}\}
	\end{equation*}
	We may again use the idea of reducing to a star tree to show the inconsistency of parsimony in this case: consider two unrooted topologies $\overline{\psi}_1, \overline{\psi}_2$ with shapes as shown in Figure \ref{fig:6taxa_shapes}. The expected costs that arise when $\sigma$ is a star tree amount to 
	
	\begin{equation*}
		\gencost{C}{\overline{\psi}_2}{\star} = \frac{43}{10} < \frac{13}{3} =  \gencost{C}{\overline{\psi}_1}{\star}
	\end{equation*}
	\begin{figure}[h]
		\centering
		\includegraphics[width=10cm]{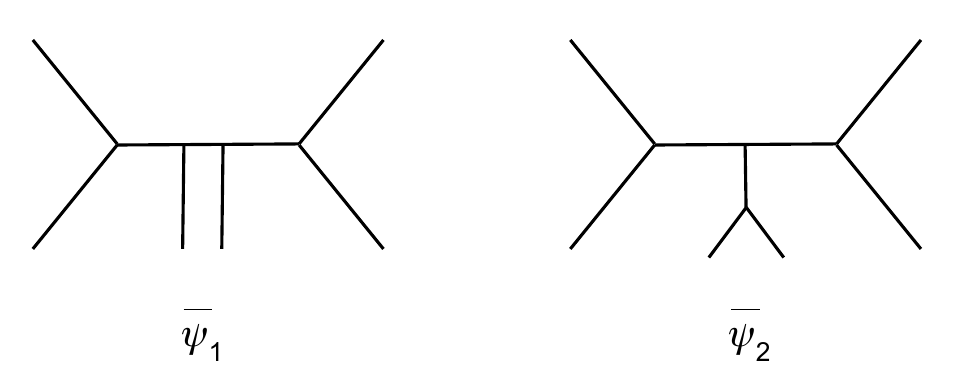}
		\caption{The two possible unrooted binary tree shapes for 6 taxa. Labels are omitted: $\overline{\psi}_1$ and $\overline{\psi}_2$ may be taken to be any unrooted labeled topologies with these respective shapes.  }\label{fig:6taxa_shapes}
	\end{figure}
	
	Therefore, parsimony fails to be statistically consistent for the unrooted 6-taxa case; if the true unrooted topology is $\overline{\psi_*}=\overline{\psi_1}$, then parsimony will always prefer one of the 15 possible unrooted topologies with shape agreeing with $\overline{\psi}_2$ for sufficiently short branch lengths $\bm{x}$. This preference was already known \citep[Equation 5]{roch2015likelihood}: in particular, it has been demonstrated that under a JC69 model of mutation, the expected difference in parsimony score per locus between $\overline{\psi}_1$ and $\overline{\psi_2}$ in the limit of a star tree is $\frac{\theta}{60} + O(\theta^2)$, where $\theta/2$ is the scaled mutation rate. This matches with our result, since we have
	\begin{equation*}
		\frac{\theta}{2} \cdot \left[\gencost{C}{\overline{\psi}_1}{\star} - \gencost{C}{\overline{\psi}_2}{\star} \right] = \frac{\theta}{2} \cdot \frac{1}{30} = \frac{\theta}{60}.
	\end{equation*}
	To visualize where exactly parsimony fails, we work with the six possible shapes of rooted $6$ taxa topologies. For each, we may in theory compute the number of PAGTs for each for any given choice of branch lengths $x_1, x_2, x_3, x_4$. However, since we have four degrees of freedom in choosing these branch lengths, performing a full analysis of parameter space (say by fixing two of the $x_i$ and varying the other two) to get a plot analogous to Figure \ref{fig:PAGTS_5taxa} is impractical. Instead, we create a plot of $K(\sigma)$ in the species case where all branch lengths are identical ($x_i=x$ for $i=1,2,3,4$ and $x \geq 0$) for each of the six shapes of topologies. The result is given in Figure \ref{fig:6taxa_pagts}.

	\begin{figure}[h]
		\centering
		\includegraphics[width=12cm]{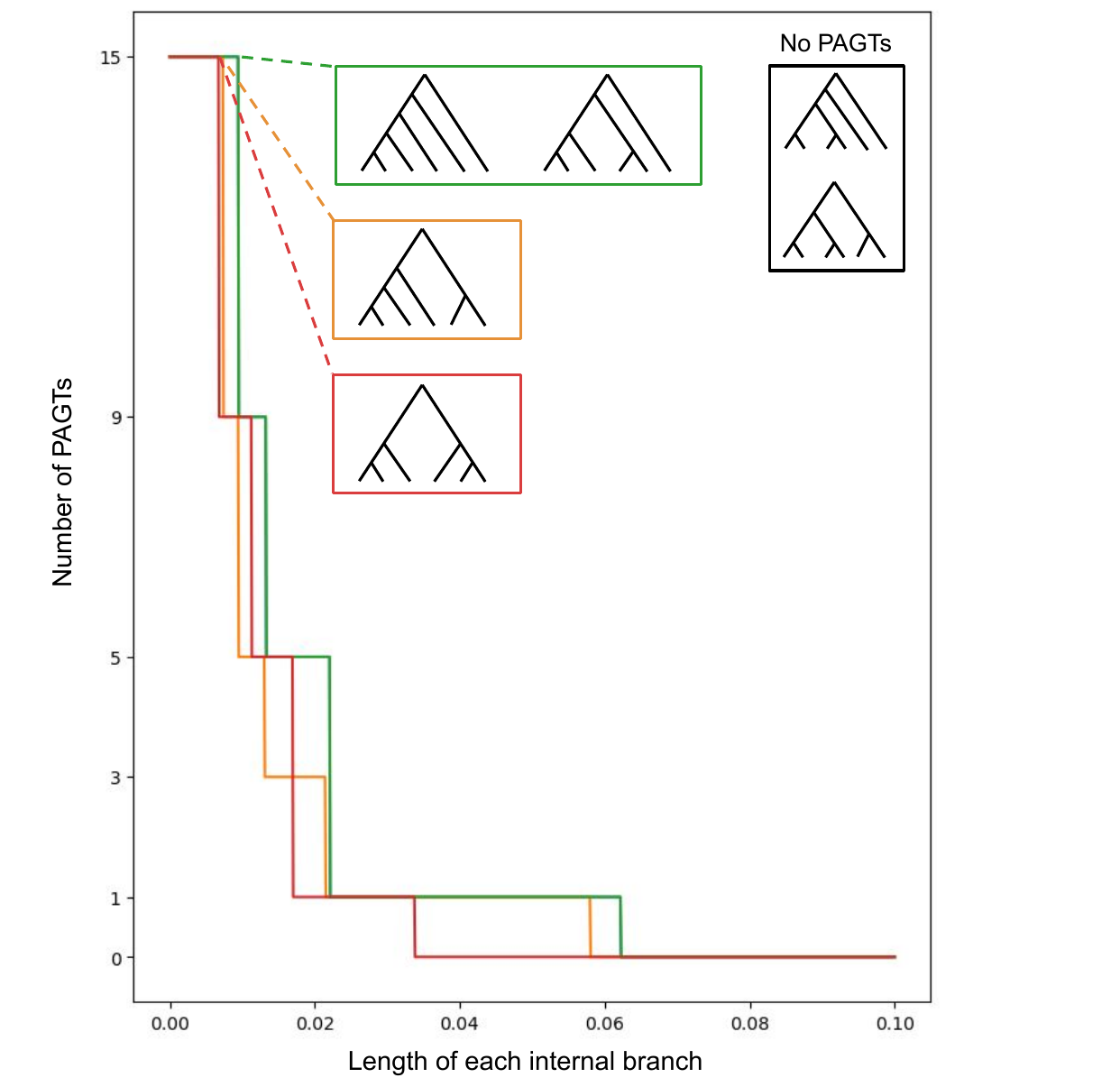}
		\caption{The number of parsimony anomalous gene trees (PAGTs) (vertical axis) for each possible rooted 6-taxa topology up to labeling, when all branch lengths have the same length (horizontal axis, coalescent units). Four topologies have any PAGTs, while two do not.}\label{fig:6taxa_pagts}
	\end{figure}
	We can again see that, identically to the rooted 5-taxa case, once the minimum branch length exceeds $x_{\rm min} \approx 0.062205$, unrooted parsimony is guaranteed to be consistent for any unrooted 6-taxa topology. The numerical agreement of the $x_{\rm min}$ needed for consistency between the rooted 5-taxa and unrooted 6-taxa case is perhaps not terribly surprising: we conjecture that such a result holds true between the rooted $n$-taxa and unrooted $(n+1)$-taxa cases for all $n\geq 5$. 
	\section{Discussion}
	
	Prior to the publication of \cite{kubatko2007inconsistency}, maximum likelihood (ML) analyses of concatenated datasets dominated phylogenetics. However, the demonstration in \cite{kubatko2007inconsistency} that concatenated ML was inconsistent when ILS was high caused a huge explosion of research into methods that are robust to gene tree discordance (e.g. ASTRAL \citep{mirarab2014astral, zhang2018astral}, MP-EST \citep{liu2010maximum}, STAR \cite{liu2009phylogenetic}). Rather than concatenate all loci, these methods instead consider each gene tree separately. Despite theoretical guarantees of consistency, such gene tree-based methods may suffer because of errors in inferring individual tree topologies from short sequences \citep{molloy2018include}. Because longer (and likely therefore concatenated) alignments offer several advantages over shorter sequences (reviewed in \citet{bryanthahn2020}), there is still a desire for concatenation methods that are robust to ILS.
	
	Results in both \cite{liu2009phylogenetic} and \cite{MendesHahn2018} found that concatenated parsimony under an infinite-sites model was consistent when applied to rooted 4-taxa trees, and we have confirmed the consistency of concatenated parsimony in the unrooted 5-taxa case under the same conditions. While concatenated parsimony had shown to be inconsistent when applied to unrooted trees with 6+ taxa in \cite{roch2015likelihood}, it was not entirely clear if this result was applicable to any biologically realistic species trees, as argued in \cite{bryanthahn2020}. The results presented here confirm that concatenated parsimony is inconsistent for a small, but non-trivial region of tree space for 5 or more taxa in the rooted case, or 6 or more taxa in the unrooted case. In this sense, the results of consistency for concatenated parsimony in the rooted 4-taxa/ unrooted 5-taxa cases appear to be a coincidence owing to the both the low-dimensionality of tree space and simplicity of the parsimony cost function in these scenarios. For instance, in the unrooted 5-taxa case, there is only one possible species tree topology up to relabeling, and to infer this topology it suffices to find the quartet $ab|cd$ with the longest internal branch length in the unrooted species tree. Concatenated parsimony does this successfully by determining the quartet $ab|cd$ supported by the most informative site patterns. 
	
	Although directly applying parsimony no longer appears to be a viable option for consistent inference from concatenated data, there are other options for statistically consistent estimation of the species tree topology. Since there is no anomalous region for parsimony in the unrooted 4-taxa case under a MSC + infinite-sites model of mutation \citep{MendesHahn2018, molloy2022theoretical}, it is possible to use parsimony to estimate the quartet $ab|cd$ for each choice of four taxa $a,b,c,d$, and then to find a tree that agrees with the greatest number of inferred trees, which is the approach taken by the methods SDPQuartets and ASTRAL-BP \citep{springer2020ils}.  Alternatively, for more general models of mutation, one may use SVDQuartets \citep{chifman2014quartet, chifman2015identifiability} to infer quartets, but this requires direct iteration through a significant portion of the $\binom{n}{4}$ possible choices of four taxa in the $n$-taxa case to guarantee accurate reconstruction, which may become prohibitively computationally expensive for large $n$. A comparable approach is CASTER \citep{zhang2025caster}, a generalization of the concatenated counting methods examined in this paper to more general models of mutation. CASTER estimates quartets in a similar manner to parsimony, though with the addition of a negative weight to some parsimony uninformative site patterns. One major advantage of CASTER over SVDQuartets is that the optimization step of CASTER avoids listing all $\binom{n}{4}$ possible quartets. 
	
	Meanwhile, concatenated distance methods, which only require iterating through the $\binom{n}{2}$ pairs of taxa in the concatenated alignment, have shown great theoretical promise for the consistent estimation. Both \cite{liu2009phylogenetic} and \cite{MendesHahn2018} found that concatenated neighbor joining (NJ) was consistent on rooted 4-taxa trees in the presence of ILS, and positive theoretical results regarding the consistency of concatenated distance methods were later examined in \cite{dasarathy2015distance} (which proposed the method METAL) and \cite{allman2019species} (which examined the use of the log-det distance). Despite these encouraging results, a recent study has found that both SVDQuartets and METAL do not perform as well as concatenated ML when using similar amounts of data in a real-life avian  dataset (\cite{braun2024testing}). As an alternative to SVDQuartets and concatenated distance methods, the mixtures across sites and trees (MAST) model of \cite{wong2024mast} has all of the advantages of concatenated ML, but allows the alignment to come from a set of alternative topologies. This approach has been found to be consistent in simulations, but more theoretical work is needed to prove its consistency more broadly.
	
	In order to examine the regions of inconsistency for concatenated parsimony, we have introduced a new mathematical method for estimating the total length of branches subtending a given sub-tree over a large number of independent and identically distributed gene trees. Under an infinite-sites model, this total length is proportional to the expected number of the corresponding informative site pattern in a concatenated alignment across all gene trees \citep{MendesHahn2018}. While this method may have further uses, it is important to point out that its relevance here depends on the infinite-sites assumption, though not necessarily on the assumption of a strict molecular clock. Future work using alternative mutation models (e.g. Jukes-Cantor) may allow our approach to be used in a wider range of scenarios and to be compared directly with work using alternative approaches \citep{roch2015likelihood, zhang2025caster}.

	\section*{CRediT authorship contribution statement}
	\textbf{D.A. Rickert}: Methodology, Software, Formal Analysis, Investigation, Visualization, Writing - Original Draft, Writing - Review \& Editing.
	\textbf{W.-T. Fan}: Writing - Review \& Editing, Formal Analysis, Funding acquisition.
	\textbf{M.W. Hahn}: Conceptualization, Methodology, Writing - Original Draft, Writing - Review \& Editing, Supervision, Project Administration, Funding acquisition.
	
	\section*{Acknowledgements}
	\noindent This work was supported by National Science Foundation grants DBI-2146866, DMS-2152103, DMS-2348164. We thank Dr. Noah Rosenberg and two anonymous reviewers whose comments helped improve this work. 
	
	\section*{Conflicts of interest}
	\noindent The authors declare no conflicts of interest. 
	
	\section*{Data availability}\label{S:data}
	\noindent Code used in the creation of quantitative figures and results is available through the folowing GitHub repository: \url{https://github.com/darickert/exp-branch-lengths}

\end{document}